\documentclass[prl,superscriptaddress,aps,notitlepage,10pt,twocolumn]{revtex4-1}
\pdfoutput=1
\PassOptionsToPackage{usenames,dvipsnames}{color}
\usepackage{amsmath,amsfonts,amssymb,amsthm,graphicx,bbm,enumerate,times,color}
\usepackage{mathtools}

\usepackage[colorlinks=false]{hyperref}

\usepackage[normalem]{ulem}

\newtheorem{theorem}{Theorem}
\newtheorem{proposition}{Proposition}

\newtheorem{lemma}[theorem]{Lemma}
\newtheorem{corollary}[theorem]{Corollary}
\newtheorem{definition}[theorem]{Definition}

\newcommand{\thestate}{\Psi^\epsilon}

\definecolor{christian}{rgb}{.8,.3,0}

\definecolor{ingo}{rgb}{.9,.2,0}

\definecolor{henrik}{rgb}{1,0,.5}

\definecolor{cadmiumgreen}{rgb}{0.0, 0.42, 0.24}

\newcommand{\mc}[1]{\mathcal{#1}}

\newcommand{\mb}[1]{\mathbb{#1}}

\newcommand{\Var}{\mathrm{Var}}

\newcommand{\e}{\mathrm{e}}
\newcommand{\ii}{\mathrm{i}}
\newcommand{\rmd}{\mathrm{d}}
\newcommand{\tr}{\mathrm{Tr}} 
\newcommand{\Tr}{\mathrm{Tr}} 

\newcommand{\id}{\mb{1}}

\newcommand{\one}{\mathbbm{1}}
\newcommand{\1}{\mathrm{id}}

\newcommand{\CC}{\mb{C}}

\renewcommand{\1}{\id}

\newcommand{\norm}[1]{\left\Vert #1 \right\Vert}

\newcommand{\ket}[1]{\left.\left|{#1}\right.\right\rangle}

\newcommand{\bra}[1]{\left.\left\langle{#1}\right.\right|}

\newcommand{\braket}[2]{\left\langle #1 \middle| #2 \right\rangle}

\newcommand{\ketbra}[2]{\ket{#1} \!\! \bra{#2}}

  \newcommand{\proj}[1]{\ketbra{#1}{#1}}

\newcommand{\dd}{\rmd}

\newcommand{\ethz}{Institute for Theoretical Physics, ETH Zurich, 8093 Zurich, Switzerland}

\begin{document}
\title{Entanglement-ergodic quantum systems equilibrate exponentially well}

\author{H.\ Wilming}
\affiliation{\ethz}
\author{M.\ Goihl}
\affiliation{Dahlem Center for Complex Quantum Systems, Freie Universit{\"a}t Berlin, 14195 Berlin, Germany}
\author{I.\ Roth}
\affiliation{Dahlem Center for Complex Quantum Systems, Freie Universit{\"a}t Berlin, 14195 Berlin, Germany}
\author{J.\ Eisert}
\affiliation{Dahlem Center for Complex Quantum Systems, Freie Universit{\"a}t Berlin, 14195 Berlin, Germany}
\begin{abstract}
One of the outstanding problems in non-equilibrium physics is to precisely 
understand when and how physically relevant observables in many-body systems equilibrate 
under unitary time evolution. 
General equilibration results show that equilibration is generic provided that the initial state has overlap with sufficiently many energy levels. 
But results not referring to typicality which show that natural initial states actually 
fulfill this condition are lacking. In this work, we present stringent 
results for equilibration for systems in which R\'enyi entanglement entropies in energy eigenstates with finite energy density are extensive for at least some, not necessarily connected, sub-system. 
Our results reverse the logic of common arguments, in that we derive equilibration from 
a weak condition akin to the eigenstate thermalization hypothesis, which is usually attributed to thermalization in systems that are assumed to equilibrate in the first place. 
We put the findings into the context of studies of many-body localization and many-body scars. 
\end{abstract}
\maketitle

Over recent years the study of the relaxation to equilibrium of complex many-body systems has attracted great attention. 
This interest can be motivated from at least two points of view. 
From a foundational viewpoint, it is desirable to understand how statistical equilibrium ensembles emerge within the framework of unitary quantum mechanics -- without introducing any external probability measures.  
It is then necessary to first explain how systems undergoing unitary evolution attain equilibrium at all.  
A key ingredient to explain this behavior has been found to be the dynamical build-up of entanglement from low-entangled initial states and results showing equilibration under quite general conditions
have been derived
\cite{Tasaki1998,Reimann2008,Linden2009,Goldstein2010,Short2011,Reimann2012,Reimann2012a,Short2012,Masanes2013,Gogolin2016,EisFriGog15,PhysRevLett.118.190601}. The increase of entanglement over time is a generic feature of complex quantum systems and leads to an increase of the entropy of subsystems over time reminiscent to the second law of thermodynamics. 

From a more concrete perspective, the recent interest in the study of
non-equilibrium dynamics is motivated by the fact that such dynamics can now be
realized in well-controlled experiments, for example in ion traps or optical lattices \cite{Bloch2008,Bloch2012,Schneider2012,Trotzky2012,Braun2015,Schreiber2015}. 
Moreover, the discovery of many-body localized systems \cite{Nandkishore2015},
which equilibrate \cite{Serbyn2014} but fail 
to thermalize \cite{Schreiber2015}, shows that there 
remains much to be understood about the equilibration 
behavior of complex quantum systems. 
Despite the great progress in understanding the equilibration behavior of many-body systems, rigorous results showing that systems with natural initial states equilibrate to high precision based on concrete physical properties have been lacking. 
 
In this article we aim to fill this gap, by taking a new perspective to the problem. 
To do this, we carefully reconsider the entanglement content of energy eigenstates in complex, interacting many-body systems and devise a working-definition of ``entanglement-ergodic'' systems whose energy eigenstates at finite energy density have a sufficient amount of entanglement between suitable subsystems. The condition we propose is very weak --  yet we show that generically such systems equilibrate to exponential precision in the volume of the system if the initial state is given by a product state with finite, non-zero energy density.
Commonly, one assumes equilibration and invokes the \emph{eigenstate thermalization hypothesis} (ETH) \cite{Deutsch1991,Srednicki1994,Rigol2008,Polkovnikov2011,Rigol2016}
to make thermalization plausible. Here, we stringently derive equilibration from a highly plausible
condition similar to, but we believe much weaker than the eigenstate thermalization hypothesis.

The main ingredient of our proof is a careful discussion of R\'enyi entanglement entropies in energy eigenstates with finite energy density. 
Combining this insight with the strongly peaked energy distribution of weakly correlated states and the monotonicity of R\'enyi entropies allows us to prove that experimentally accessible initial states are well smeared out over the energy spectrum, which implies high-precision equilibration for generic interacting Hamiltonians. 

\paragraph{Formal setting.}
We consider local Hamiltonians
\begin{equation}
	H_\Lambda=\sum_{x\in\Lambda} h_x
\end{equation}
on a regular lattice $\Lambda$ in $\nu$ spatial dimensions with $N:=  |\Lambda|$ lattice sites.
The Hilbert space is $\otimes_{x\in\Lambda} \mc H_x$ with $\operatorname{dim}\mc H_x=d$.
Since we will be talking about the scaling of quantities with the lattice size, $H_\Lambda$ should be seen as a sequence of Hamiltonians, which is, for example, given by a family of translational invariant system on larger and larger square lattices with periodic boundary conditions. 
We call $H_\Lambda$ $l$-local if the diameter of the support of each $h_x$ is at most $l$,  strictly local if it is $l$-local with $l$ independent of the system size $N$ and uniformly bounded if $\norm{h_x}\leq h$ for all $x\in\Lambda$ independent of the system size.
Since we are mostly interested in energy densities instead of total energies
later, we make the unusual choice to label eigenvectors of the Hamiltonian by their energy 
\emph{densities} $e_i$ as $\ket{e_i}_\Lambda= \ket{E_i/N}_\Lambda$, with $i=1,\ldots,d^N$ and $E_i$ being the eigenvalue of $H_\Lambda$ corresponding to $\ket{e_i}_\Lambda$.
We always assume that the ground-state has vanishing total energy.
From now on, we will often drop the subscript $\Lambda$ from states and Hamiltonians to simplify the notation.

\emph{Equilibration in closed systems.}
We now briefly review general equilibration results that we will be using in the following.
Consider any bounded observable $A$ and any initial state $\rho$.  
Denote the time evolved states by $\rho(t)$. We want to study whether the expectation value $\langle A(t)\rangle_\rho = \tr(\rho(t)A)$ equilibrates over time.  
In a finite system, perfect equilibration, in the sense that $\langle A(t)\rangle$ becomes static for all times after the equilibration process, is impossible due to recurrences. 
However, it is perfectly possible that this value is very close to a stationary value $\overline{A}$ for most of the time, with rare deviations. 
The value $\overline{A}$ is then necessarily 
the 
infinite time average
\begin{align}
\overline{A} :=  \overline{\langle A(t)\rangle} = \lim_{T\to \infty}\frac{1}{T}\int_0^T \tr(\rho(t)A) = \tr(\omega A),
\end{align}
where $\omega$ denotes the time average $\overline{\rho(t)}$, which is again a valid density matrix. 
In the case where the Hamiltonian has no degenerate \emph{energy differences} $G_{(i,j)}=E_i-E_j$, it has
been shown that the time-averaged fluctuations around the equilibrium value are bounded by \cite{Short2011} 
\begin{align}\label{eq:eqonaverage}
\Var(A,H,\rho):=  \overline{(\langle A(t)\rangle - \overline{A})^2} \leq \norm{A}^2 \e^{-S_2(\omega)}, 
\end{align}
where $S_\alpha$ denotes the R\'enyi-$\alpha$ entropy
\begin{align}
S_\alpha(\rho) :=  \frac{1}{1-\alpha}\log\left(\tr(\rho^\alpha)\right).
\end{align}
We note for later use that in the limit $\alpha \rightarrow 1$ the von-Neumann entropy is
recovered and that R\'enyi entropies are monotonically decreasing in $\alpha$.
The condition of non-degenerate energy differences is generically fulfilled in interacting systems \cite{symmetryNote}.
However, generalizations of the above result also exist if this condition is weakly violated \cite{Short2012}.
It is also possible to show $\Var(A,H,\rho)\leq 3\norm{A}^2 \e^{-S_\infty(\omega')}$, where $\omega'$ is the same operator as $\omega$ but with its largest eigenvalue replaced by zero \cite{Reimann2008}, which sometimes gives a stronger bound. 
In particular, it also incorporates the case of energy eigenstates, which are always fully equilibrated.
It is important to stress, however, that the bound Eq.\ \eqref{eq:eqonaverage} does not lead to implications on the time it takes to observe equilibration.

Similarly to results in terms of bounded observables, one can also bound the distance of a local reduced state on a subsystem $S$ from its time average as \cite{Linden2009}
\begin{align}\label{eq:eqsubsystem}
	 \overline{\mc \| \tr_{S^c}(\rho(t)) - \tr_{S^c}(\omega)\|_1 } \leq 2 d_S \e^{-S_2(\omega)/2}.
\end{align}  
Here, $\|.\|_1$ denotes the trace norm, which bounds the difference in expectation value of all normalized
observables as 
\begin{align}
	\|\rho-\sigma\|_1=\max_{A,\norm{A}=1} |\tr(\rho A) - \tr(\sigma A)|.
\end{align}
Roughly speaking, the role of the norm of the observables in the previous bound is here taken by the dimension of the subsystem.  
We conclude that a large R\'enyi-2 entropy of the time-averaged state is a sufficient condition for a generic, closed quantum system to equilibrate eventually.
Unfortunately, however, there are few general and rigorous results which show
that natural initial states lead to time-averaged states whose energy
distribution has a large R\'enyi-2 entropy \cite{Farrelly2016,Gallego2017}.
Ref.~\cite{Gallego2017} makes it highly plausible from an operational
point of view that states that can be prepared in experiments have large
effective dimension, but does not show it for concrete states. 
Ref.~\cite{Farrelly2016} shows that states with finite correlation length have R\'enyi-2 entropy of at least roughly the order $\log(N)$. 
While this formally leads to equilibration as $N\rightarrow \infty$, 
it is insufficient to obtain a finite entropy density at equilibrium (which is crucial from a thermodynamic point of view) and requires very large system sizes to explain equilibration.
In the following, we show that in systems which have a sufficiently large
amount of entanglement in energy eigenstates with finite energy density, finite
entropy density and hence exponentially good equilibration in the system size follows for initial product states.

\emph{Entanglement and R\'enyi entropies.}
Here, we are interested in ergodic, non-integrable systems. 
Contrary to the case of classical mechanics, there is no generally agreed upon definition of what it exactly means for a quantum system to be ergodic or non-integrable \cite{Caux2010}. 
In recent years, it has been argued that a general characteristic of systems which can be considered ergodic is that energy eigenstates fulfill a so-called \emph{volume law} in terms of their entanglement content. 
The condition we propose is inspired from this observation, but is much weaker. 
It is therefore useful to discuss volume laws before stating our condition.
Consider an energy eigenvector with finite energy density $\ket{e}$ and denote by $\rho_A(e)$ the reduced density matrix on some (contiguous) subsystem $A$ which is smaller than one half of the total system, but still contains a finite fraction of the total system.  A \emph{volume law} means that the entanglement measured by a R\'enyi entropy $S_\alpha$ grows like the volume of $A$ as
\begin{align}
S_\alpha(\rho_A(e)) \sim |A|. 
\end{align}   
A natural question to ask is for which value of $\alpha$ this relation is supposed to hold. 
We now argue that this relation is a meaningful criterion only if $\alpha>1$.
This might come as a surprise since it is common to measure entanglement in terms of the von-Neumann entropy $S_1$. 
This is due to the fact that the von-Neumann entropy describes the fraction of EPR-pairs that can be distilled from asymptotically many copies of a state by local operations and classical communication \cite{PhysRevA.53.2046}. 
However, the examples presented in the following proposition show that a volume law in terms of von-Neumann entropy is not a very useful criterion to determine whether a state of a many-body system deviates strongly from an unentangled state.

\begin{proposition}[Counterexample]\label{prop:example}
	For any $1>\epsilon>0$ there exist state vectors $\ket{\thestate}$ on $\Lambda$ with the following properties:
	\begin{enumerate}
		\item\label{prop:product} $\ket{\thestate}$ has overlap exponentially close to $1-\epsilon$ in 
		$N$ with a product state vector $\ket{\Psi}$.
	\item $\ket{\thestate}$ fulfills a \emph{volume law} in the von~Neumann entropy: 
	There exist regions $A$ with $|A|=N/2$ such that
			$S_1(\rho_A) \simeq \frac{\epsilon}{2} \log(d) N$.
	\item\label{prop:renyi} All R\'enyi entropies with $\alpha>1$ are bounded by a constant in the system size,
		$S_\alpha(\rho_A) \leq \mathrm{const.}$ 
	\end{enumerate}
\end{proposition}
The state vectors that fulfill these condition are simply of the form 
	$\ket{\Psi^\epsilon}\propto \sqrt{1-\epsilon}\ket{\Psi} + \sqrt{\epsilon} \ket{\Omega}$,
where $\ket{\Omega}$ is maximally entangled between $A$ and its complement. 
The proof of these properties is given in \ref*{app:example} in the Supplemental Material. 
Given Properties~\ref{prop:product} and \ref{prop:renyi}, it cannot reasonably be said that the amount of entanglement in $\ket{\Psi^\epsilon}$ grows volume-like -- even though it fulfills a volume law in terms of von-Neumann entropy.
At the same time it suggests that we should require a volume law in terms of some R\'enyi entropy with $\alpha>1$. 
Importantly, the inequality (proven in the Supplemental Material, \ref*{app:prooflemma})
\begin{align}
		S_\alpha(\rho) \geq S_\infty(\rho) \geq \frac{\beta-1}{\beta} S_{\beta}(\rho),\quad \forall \alpha\geq 0,\label{eq:relation}
\end{align}
which holds for any $\beta>1$, shows that all the R\'enyi entropies with $\alpha>1$ have the same scaling behaviour. It therefore does not matter which one we consider and in the following we therefore mostly consider the case $\alpha=2$.

Recently there has been an increasing amount of numerical results and theoretical arguments that show that energy eigenstates of generic non-integrable quantum systems with finite energy density have this property (see, for example, Refs.~\cite{Garrison2015,Nakagawa2018,Lu2017,Huang2017,Vidmar2017,Vidmar2017a,PhysRevB.92.180202}).
These results fit well to, and indeed are partly motivated by, the observation that certain properties of complex, strongly interacting systems at finite energy density can be well described by assuming that their Hamiltonians are \emph{random matrices} \cite{Berry1985,Prosen1999,Mueller2004,Kollath2010,Keating2015,Kos2017}, despite the Hamiltonians being local {and thus belonging to a set of measure zero. 
It is well known that Haar random quantum states have extensive R\'enyi entanglement entropy on bipartite systems with very high probability \cite{Lubkin1978,Lloyd1988,Page1993,Nadal2010,Nadal2011}. 
Since eigenstates of random Hamiltonians are distributed according to the Haar-measure, we expect that eigenstates of random Hamiltonians fulfill a volume law in terms of all R\'enyi entropies with very high probability \cite{HaarNote}. 

\emph{Entanglement-ergodicity.} In essence, in the following we will show that an extensive amount of entanglement entropy in energy eigenstates for R\'enyi-entropies with $\alpha>1$ is a sufficient criterion for exponentially good equilibration. 
Importantly, however, we will significantly weaken the assumption from volume laws -- where the subsystem under consideration is assumed be a contiguous region -- by allowing the subsystem to be almost completely arbitrary. 
The only property that we demand from the subsystem is that it includes a finite fraction of the total system. 
In particular, it need not be connected and further may be chosen differently for every energy eigenstate. 
For example, it may have a fractal-like shape (up to the lattice spacing), or consist of a sub-lattice of spins that are far away from each other when compared to some natural length scale of the system.
In the following, we call this (much weaker) form of a volume law a \emph{weak volume law}.
In particular, even states that have a finite correlation length, such as \emph{matrix product states (MPS)}, can be expected to generically follow a weak volume law -- unless they are product states themselves. 

We will now give a formal definition of what we demand from an entanglement-ergodic system. Notably, this definition is perfectly compatible with the mindset of the ETH that argues basically that eigenstates of local Hamiltonians in the bulk of the spectrum should be locally indistinguishable from Gibbs states due to their entanglement.  

\begin{definition}[Entanglement-ergodicity]
We call a sequence of systems of increasing system size \emph{entanglement-ergodic}, if there exists a system size $N_0$ and a function $g:\mathbb{R}\rightarrow [0,\infty)$, 
such that for all system sizes $N\geq N_0$ it holds that	
\begin{enumerate} 
	\item (Weak volume law) For every eigenvector $\ket{e_i}_\Lambda$ there exists \emph{some} subsystem
$A_\Lambda$ such that the reduced state
$\rho_{A_\Lambda}(e_i):= \tr_{A_\Lambda^c}(\proj{e_i}_\Lambda)$ fulfills
\begin{align}
	S_{2}(\rho_{A_\Lambda}(e_i)) \geq g(e_i) N.
\end{align}
\item The function $g$ is sufficiently well-behaved: It is Lip\-schitz continuous and positive for non-extremal energy densities, i.e., $g(e)>0$ for $ 0<e<e_\mathrm{max}$.
\end{enumerate}
\end{definition}
Some remarks on this definition are in order: 
i) As shown by \eqref{eq:relation}, we could have replaced the R\'enyi-2 entropy with any R\'enyi entropy with $\alpha>1$ and would have obtained an equivalent definition. 
ii) We only require that there is a \emph{lower bound} on the entropy in the region $A_\Lambda$ as a function of the energy density and not that that the entropy is given by the function $e\mapsto g(e)$. 
In particular, we do not require that eigenstates with the same energy density also have the same R\'enyi entanglement entropy and allow for the possibility that states with vanishing energy density do not have a large amount of entanglement. 
For example, the ground state of the system may be a product state. 
Similarly, we impose Lip\-schitz continuity of the function $g$ for simplicity and concreteness. 
Similar conclusions as in the following can be reached by imposing different regularity assumptions on $g$.
iii) It would be perfectly fine for all what follows to allow for additional negative but sub-leading terms, e.g.\ in $O(\sqrt{N})$. 
These terms would only change the sub-leading behavior of our final results and may encode non-asymptotic information about the entropy in eigenstates. 
For simplicity, we omit such terms here. 

iv) We emphasize again that we only demand a weak volume law, i.e., that for any system size and any eigenstate there exists \emph{some} finite fraction of the total system $A_\Lambda$ whose entropy is sufficiently large. It is not required that this holds for all such subsystems nor is it required that this subsystem has any particular shape. 
For example, it could consist of every 10-th site of the lattice.
Nevertheless, in generic, strongly interacting systems, we expect that choosing $A_\Lambda$ simply as one half of the system is sufficient.

v) Importantly, the weak assumption on the subsystem $A_\Lambda$ allows even MPS-like eigenstates, which fulfill an \emph{area law} \cite{AreaReview} for \emph{contiguous} regions, to be entanglement-ergodic. As a concrete example, we prove the following statement in the Supplemental Material, \ref*{app:finitecorrelationlength}:\emph{
	States that are prepared by a translationally invariant, finite-depth, local quantum circuit and are not a product state have extensive in R\'enyi-$2$ entropy on a finite fraction of the system.}
Similarly, we expect an analogous result to hold more generally for generic MPS. A detailed discussion on the application of our framework to generic MPS is provided in \ref*{app:finitecorrelationlength} of the Supplemental Material.  
These results show that also systems featuring \emph{many-body localization} \cite{Nandkishore2015}, whose eigenstates are expected to be approximable by matrix product states \cite{Bauer,1409.1252}, may fall within the framework of entanglement ergodicity. 
It is well known that such systems equilibrate, but fail to also thermalize. 
Thus even if we refer to the key property as entanglement-ergodicity, it is a significantly weaker condition
than what is commonly understood as ergodicity in current literature.

vi) Similarly, we expect the notion of entaglement-ergodicity to be weaker than the ETH: if one assumes that the ETH applies to some subsystems containing a finite fraction of the total system, then we strongly believe that the ETH implies entanglement-ergodicity and hence equilibration with high precision (see Supplemental Material, \ref*{app:eth}).

vii) With stronger assumptions on the regions $A_\Lambda$, we can extend the applicability of the definition of entanglement ergodicity to states related by quasi-local unitaries (those generated by time-evolution under local Hamiltonians): \emph{For systems in which the regions $A_\Lambda$ have an asymptotically vanishing surface-to-volume ratio, entanglement-ergodicity is stable under quasi-local unitaries.}
The precise meaning and formal proof of this statement is formulated 
in \ref*{app:quasilocal} of the Supplements.
Intuitively it follows by observing that time evolution under a local Hamiltonian for a finite time can only decrease the entropy of a sub-region by an amount that is proportional to its boundary, since the entropy has to  ``flow'' 
out of the subregion through its boundary. 

\emph{Consequences of entanglement-ergodicity.}
Let us now discuss the consequences of entanglement-ergodicity. We first state a result on the 
diagonal entropy in entanglement-ergodic systems, to then turn to the implication that
entanglement-ergodic systems equilibrate exponentially well.

\begin{theorem}[Diagonal entropy in entanglement-ergodic systems]\label{thm:main}
	Consider an entanglement-ergodic system with strictly local, uniformly bounded Hamiltonian. Then for any energy density $e>0$ there exists a constant $k(e)>0$ and a system-size $N_0(e)$ such that for all system-sizes $N>N_0(e)$ and for all product-states $\ket{\Psi}_\Lambda$ with energy density $e$, we have 
\begin{eqnarray}\label{eq:main:quasi}
	S_\alpha(\omega_\Lambda) \geq k(e) N,
\end{eqnarray}
	where $\omega_\Lambda$ is the time-average of $\proj{\Psi}_\Lambda$.
\end{theorem}
As a direct consequence of this result we obtain from \eqref{eq:eqonaverage} the following bounds on equilibration.
\begin{corollary}[Equilibration in entanglement-ergodic systems]
Under the same conditions as in Theorem~\ref{thm:main} and the additional assumption of non-degenerate 
	energy gaps of the Hamiltonian, there exists a constant $k(e)>0$ such that
\begin{align}
	\Var\left(A,H_\Lambda,\proj{ \Psi}_\Lambda\right) \leq \norm{A}^2 \e^{-k(e) N}.
\end{align}
\end{corollary}
Similar bounds hold for the reduced state on a small sub-system, as implied by \eqref{eq:eqsubsystem}. 
The proof of Theorem~\ref{thm:main} is given in \ref*{app:prooftheorem} of the Supplemental Material.
It relies on recognizing that an extensive amount of R\'enyi entanglement entropy implies that a state has exponentially small overlap with all product states 
and to combine this statement with a recent central limit-type theorem for the energy distribution in product states \cite{Anshu2016}.
The premises of Theorem~\ref{thm:main} require the initial state to be a product states. In the Supplements (\ref*{app:state_quasilocal}), we further extend our results to the case where the initial state is prepared from a product state by a quasi-local unitary under the additional assumption that the regions $A_\Lambda$ have vanishing surface-to-volume ratio. 
We obtain equilibration bounds scaling as $O(\exp(-N^\beta))$ for some $0<\beta<1$.

\emph{Quantum many-body scars.} Recently, it has been observed that sets of atypical energy eigenstates with small amount of entanglement may rarely show up even at finite energy density in non-integrable, kinetically constrained many-body systems -- a phenomenon dubbed \emph{``quantum many-body scars''}, 
which leads to exceedingly slow equilibration with long-lived oscillations from certain initial product states \cite{Bernien2017,Turner2018,Moudgalya2018,Turner2018a,Moudgalya2018a}.  
Even a toy-model with complete absence of equilibration for certain initial product states has been constructed \cite{Choi2018}. 
This complete absence of equilibration can be explained due to a ${\mathrm{SU}}(2)$ symmetry in the subspace of many-body scars \cite{Choi2018}, which also emerges approximately in more realistic Hamiltonians and breaks the assumption of non-degenerate energy gaps. Furthermore, many-body scars may have large squared overlap with the initial product state (of order $1/N$), which implies that such systems also violate entanglement-ergodicity within the subspace of many-body scars.

\emph{Conclusion.}
We carefully formalized a notion of ergodicity based on extensive R\'enyi entanglement entropies in energy eigenstates and showed that this notion suffices to prove exponentially precise equilibration for Hamiltonians with non-degenerate energy gaps. Our condition is quite weak and we expect it to be fulfilled for generic interacting systems.
The notion of ergodicity we introduced is connected to the eigenstate thermalization hypothesis (ETH), which asserts that local reduced states of energy eigenstates with finite energy density already resemble the reduced state of a corresponding Gibbs state. 
We thereby introduce a new perspective to the study of non-equilibrium quantum systems: We do not
have to assume equilibration for systems to become apparently stationary and then turn to the ETH to show thermalization. 
Instead, we show that equilibration already follows from a weak ETH-like assumption.
In fact, in this case the ETH implies much more, since our definition does in general not imply that the system also \emph{thermalizes} (this apparent shortcoming is necessary when formulating a criterion that may also apply to many-body localized systems). 
Here, we did not discuss the time scales for relaxation to equilibrium, but were interested in the precision of equilibration after arbitrarily long times. While some progress in understanding equilibration time-scales has been made recently, both in integrable~\cite{Calabrese2006,Calabrese2007,Eisler2007,Rigol2007,Cramer2008,Cramer2008a,Barthel2008,Flesch2008,Calabrese2011,Caux2013} and generic, non-integrable systems~\cite{Short2012,Goldstein2013,Malabarba2014,Garcia-Pintos2015a,Goldstein2015,Farrelly2016a,Reimann2016,DeOliveira2017,Wilming2017}, finding rigorous arguments bounding equilibration time scales from reasonable assumptions remains an outstanding open problem.

\paragraph{Acknowledgements.} H.~W. would like to thank Rodrigo Gallego for fruitful discussions regarding the role of R\'enyi entropies in equilibration. We acknowledge financial support from the ERC (TAQ), the DFG 
(FOR2724, CRC 183, EI 519/7-1, EI 519/9-1), the Templeton Foundation, and the Studienstiftung des Deutschen Volkes. H.~W. further acknowledges contributions from the Swiss National Science Foundation via the NCCR QSIT as well as project No. 200020\_165843. This work has also received funding from the European Union's Horizon 2020	research and innovation programme under grant agreement No 817482 (PASQuanS).

%
 \phantom{x}
 \makeatletter
 \newcommand{\manuallabel}[2]{\def\@currentlabel{#2}\label{#1}}
 \makeatother
 \manuallabel{app:prooflemma}{Section~A}
 \manuallabel{app:example}{Section~B}
 \manuallabel{app:finitecorrelationlength}{Section~C}
 \manuallabel{app:prooftheorem}{Section~D}
 \manuallabel{app:quasilocal}{Section~E}
 \manuallabel{app:state_quasilocal}{Section~F}
 \manuallabel{app:eth}{Section~G}
 \clearpage

 \appendix
 \makeatletter
 \onecolumngrid
 \begingroup
 \frontmatter@title@above
 \frontmatter@title@format
 \@title{} -  Supplemental Material
 
 \endgroup
 \vspace{1cm}
 \twocolumngrid
 \makeatother
 \setcounter{page}{1}
\onecolumngrid 
 \subsection{\ref*{app:prooflemma}: Proof of Equation~(12) in main text and a Lemma on overlaps with product states}
 In this section, we provide the derivation of the inequality
 \begin{equation}\label{eq:relation}
 		S_\alpha(\rho) \geq S_\infty(\rho) \geq \frac{\beta-1}{\beta} S_{\beta}(\rho),\quad \forall \alpha\geq 0, \beta > 1.
 \end{equation}
 Furthermore, we derive a key technical lemma that we will use in multiple of the subsequent arguments. 
 \begin{proof}[Proof of Equation~\eqref{eq:relation}]
 Let $\{q_j\}$ be the eigenvalues of $\rho$ and let $q_0$ be the largest one (in case of degeneracies, pick any of the largest ones).
 Since $q_0^{\tilde\alpha} \leq \sum_{j} q_j^{\tilde\alpha}$, we have, for any $\tilde\alpha>1$,
 \begin{align}
 	\frac{\tilde \alpha}{\tilde\alpha-1}S_\infty(\rho) &= - \frac{\tilde \alpha}{\tilde\alpha-1}\log(q_0) = \frac{1}{1-\tilde\alpha}\log(q_0^{\tilde\alpha})\\
 						       &\geq \frac{1}{1-\tilde\alpha}\log\left(\sum_j q_j^{\tilde\alpha}\right) = S_{\tilde\alpha}(\rho).
 \end{align}
 Since the R\'enyi entropies are monotonically decreasing in $\alpha$, this proves the claim. 
 \end{proof}
 
 We now derive a helpful result concerning the overlap of product states with arbitrary pure states. Among others, our proof of Theorem~\ref{thm:main} on the equilibration of entanglement ergodic systems will involke this lemma. 
 
 \begin{lemma}[Overlap of product states with arbitrary pure states]\label{lemma:overlap} 
 Let $\ket{\Phi}$ be any pure state vector on a lattice, $A$ be a fixed subregion and $\sigma_A$ be the corresponding reduced 
 quantum state. 
  Let $\alpha >1$. Then
  \begin{align}
  |\braket{\Phi}{\Psi}|^2 \leq \e^{-\frac{\alpha-1}{\alpha} S_{\alpha}(\sigma_A)} 
  \end{align}
  for any product vector $\ket{\Psi} = \ket{\Psi_A} \otimes \ket{\Psi_{A^c}}$. 
 \end{lemma} 
 \begin{proof}
 To prove this Lemma, we make use of well-known properties of the fidelity
 \begin{align}
 	F(\rho,\sigma) :=  \tr\left( ( {\rho}^{1/2}\sigma{\rho}^{1/2})^{1/2}\right)
 \end{align}
 between two quantum states. 
 In the case of two pure states, this quantity is simply given by their overlap. 
 Moreover, the fidelity can only increase when we trace out a subsystem, which follows from Uhlmann's theorem \cite{Uhlmann1976,Nielsen2000} as
 \begin{align}
 	F(\rho,\sigma) \leq F(\rho_A,\sigma_A).
 \end{align}
 The reduced state of the product state vector $\ket{\Psi}$ is again pure. 
 We denote this state vector by $\ket{\Psi_A}$ and would like to 
 remind the reader that $\sigma_A$ is the reduced state corresponding
 to $\ket{\Omega}$.
 We then have
 \begin{align}
 	|\braket{\Psi}{\Omega}|^2 &\leq F(\proj{\Psi_A},\sigma_A)^2 \\
 	&=\bra{\Psi_A}\sigma_A\ket{\Psi_A} \leq \norm{\sigma_A}_\infty = \e^{-S_\infty(\sigma_A)}.
 \end{align}
 The claim then follows from \eqref{eq:relation}. 
 \end{proof}

 \subsection*{\ref*{app:example}: Entropy calculations for the example}
 Here, we show that $\ket{\thestate}$ fulfills the properties claimed in Proposition~\ref{prop:example}. 
 We will make use of Lemma~\ref{lemma:overlap} of \ref*{app:prooflemma}. 
 In particular, for $\alpha = \infty$ the lemma states that $|\braket{\Phi}{\psi}|^2 \leq \exp[- S_\infty(\sigma_A)]$ with $S_\infty(\sigma_A) = - \log \norm{\sigma_A}$.
 
 First let us determine the normalization of $\ket{\thestate}$:
 \begin{align}
 	\norm{\sqrt{1-\epsilon} \ket{\Psi} +\sqrt{\epsilon}\ket{|\Omega}}_2^2 &= 1 + \sqrt{(1-\epsilon)\epsilon} (\braket{\Psi}{\Omega} + c.c.).
 \end{align}
 For the maximally entangled state $\ket{\Omega}$ we have 
 $S_\infty\left(\tr_{A^c}\left(\ketbra{\Omega}{\Omega}\right)\right) = \log d_A$, where $d_A:= d^{|A|}$ denotes the Hilbert-space dimension of the region $A$.
 Since $\ket{\Psi}$ is a product state, Lemma~\ref{lemma:overlap} yields (for $\alpha=\infty$)
 \begin{align}
 	|\braket{\Psi}{\Omega}|^2 \leq \e^{-\log(d) |A|}.
 \end{align}
 Hence, the vector $\sqrt{1-\epsilon}\ket{\Psi} + \sqrt{\epsilon}\ket{\Omega}$ is normalized up to an error that is exponentially small in $A$. Since we are interested in scaling relations, we will in the following  therefore simply assume
 \begin{align}
 	\ket{\Psi^\epsilon} = \sqrt{1-\epsilon}\ket{\Psi} + \sqrt{\epsilon}\ket{\Omega}.
 \end{align}
 One may explicitly check in the calculations below that the error introduced by this assumption is irrelevant for the scaling statements that we are making.

 Let us now show that the state vector $\ket{\thestate}$ fulfills a volume law in terms of von-Neumann entropy, but not in terms of any R\'enyi entropy with $\alpha>1$.
 Since $\ket{\Psi}$ is a product, 
 its marginals on the region $A$ and its complement $A^c$ are pure and we denote the corresponding 
 state vectors as $\ket{\Psi_A}$ and $\ket{\Psi_{A^c}}$, respectively. It is useful to introduce the state vector $\ket{\Phi}$ in the Hilbert space of region $A$ given by
 \begin{align}
 	\ket{\Phi} := \frac{(\one\otimes \bra{\Psi_{A^c}})\ket{\Omega} }{\norm{(\one\otimes \bra{\Psi_{A^c}})\ket{\Omega}}_2}.
 \end{align}
 Due to Lemma~\ref{lemma:overlap}, the norm $\delta\coloneqq\norm{(\one\otimes \bra{\Psi_{A^c}})\ket{\Omega}}_2$ can be upper bounded as 
 \begin{align}
 		\delta^2 = \mathrm{max}_{\ket{\chi_A}}\left|(\bra{\chi_A}\otimes\bra{\Psi_{A^c}})\ket{\Omega}\right|^2 \leq \e^{-\log(d_A)} = \e^{-\log(d) |A|}.
 \end{align}
 With these ingredients, let us compute the reduced density matrix of the state
 $\ket{\thestate}$ on region $A$ as
 \begin{align}
 	\tr_{A^c}(\proj{\thestate}) &= (1-\epsilon)\proj{\Psi_A} + \epsilon \frac{\one}{d_A} + \sqrt{\epsilon(1-\epsilon)}\Tr_{A^c}(\ketbra{\Psi}{\Omega} + \mathrm{h.c.}) \nonumber \\
 					&= (1-\epsilon)\proj{\Psi_A} + \epsilon \frac{\one}{d_A} + \sqrt{\epsilon(1-\epsilon)}\delta(\ketbra{\Psi_A}{\Phi} + \mathrm{h.c.}),
 \end{align}
 where the second line follows by taking the partial trace in terms of a basis that includes $\ket{\Psi_{A^c}}$ as one of its elements. 
 For small $\delta>0$, the spectrum of this density matrix is well-approximated by
 the spectrum of the state $\xi_A:=  (1-\epsilon)\proj{\Psi_A} + \epsilon
 {\one}/{d_A}$, since the perturbation proportional to $\delta$ has support on a two-dimensional subspace spanned by $\Psi_A$ and the part of $\Phi$ that is perpendicular to $\Psi_A$. This perturbation thus only affects two spectral values of $\xi_A$.
 Since $\delta$ is exponentially small in the size of $A$, we simply neglect this perturbation in what follows and work with the state $\xi_A$. The von-Neumann entropy of the state $\xi_A$ can easily be calculated as
 \begin{align}
 	S_1(\xi_A) &= -(1-\epsilon+\frac{\epsilon}{d_A}) \log(1-\epsilon+\frac{\epsilon}{d_A}) + (d_A-1)\frac{\epsilon}{d_A}\log\left(\frac{\epsilon}{d_A}\right) \nonumber \\
 					       &\geq \frac{d_A-1}{d_A}\epsilon\log(d_A)\, \approx \epsilon\,\log(d) |A|,
 \end{align}
 where the last approximation is exponentially good in $|A|$. 
 We thus find a volume law in terms of the von-Neumann entropy. Let us now calculate the R\'enyi entropy for $\alpha>1$. 
 In this case we can again use the inequality
 \begin{align}
 	S_\alpha(\xi_A) \leq \frac{\alpha}{\alpha-1}S_\infty(\xi_A), \quad \alpha>1.
 \end{align}
 Since the largest eigenvalue of $\xi_A$ is given by $(1-\epsilon+\frac{\epsilon}{d_A})$, we then have
 \begin{align}
 S_\alpha(\xi_A) &\leq \frac{\alpha}{\alpha-1}\log\left(\frac{1}{1-\epsilon + \frac{\epsilon}{d_A}}\right)\\
        &\leq \frac{\alpha}{\alpha-1}\log\left(\frac{1}{1-\epsilon}\right),
 \end{align}
 where we have made use of the fact that $\log(1/x)$ is monotonically decreasing in $x$. 
 We thus find that all R\'enyi entropies for $\alpha>1$ are upper bounded by a constant independent of the size of $A$. 
 
 \subsection{\ref*{app:finitecorrelationlength}: Extensive entropy for states prepared by finite-depth circuits and MPS}
 In this appendix, we show that a state vector $\ket{\Psi}$ that is prepared by a translationally invariant, finite-depth, local quantum circuit has either extensive R\'enyi-2 entropy on a suitable sub-system or is a product state itself.
 By a finite-depth, local quantum circuit we mean a unitary operator $U$ that can be written as a concatenation of $D$ unitaries which consists of a tensor product of unitary operators, each of which acts on spins that are separated by a lattice distance at most $k$ and we assume that $D$ does not depend on the system size.
 By translationally invariance, we simply mean that the final state is invariant under shifts of the underlying lattice by some finite lattice vector. We assume this to simplify the discussion, but expect similar results to hold for generic such states. 
 We note that states prepared by a local, finite depth quantum circuit are a special sub-class of MPS. We discuss the case of general MPS further below.
 
 The crucial property that we will make use of is that $U$ has a strict ``light-cone'' for the spreading of correlations. This implies that spins that are separated by a distance more than $2kD+1$ are uncorrelated. 
 Let us hence consider as subsystem $A$ a sub-lattice $\tilde \Lambda$ of spins that are separated by at least a distance $2kD+1$. Then the reduced state on $A$ is a product state:
 \begin{align}
 	\rho_A = \bigotimes_{x\in\tilde\Lambda} \rho_x.
 \end{align}
 By translational invariance, we furthermore ensure that all the $\rho_x$ are identical, i.e.\ $\rho_x = \varrho$. Since the R\'enyi entropies are additive over tensor-products, we thus find that
 \begin{align}
 	S_2(\rho_A) = |\tilde \Lambda| S_2(\varrho) \sim N S_2(\varrho).
 \end{align}
 Note that by assumption $\varrho$ does not depend on $N$. Hence, if the state has a non-extensive amount of entropy it must have $S_2(\varrho)=0$ and in consequence the state $\varrho$ must be pure. By repeating this argument for different sub-lattices $\tilde \Lambda$ we conclude that every spin is in a pure state and hence $\ket{\Psi}$ is a product state. In conclusion, $\ket\Psi$ has either extensive R\'enyi-2 entropy or is a product state. 
 
 Let us now turn to the case of generic matrix product states. For simplicity, we again consider here the case of translationally invariant MPS, which take the form
 	\begin{align}
 		\ket{\mathrm{MPS}[A]_N} := \frac{1}{\sqrt{Z}}\sum_{i_1,\ldots i_N=1}^d \Tr(A_{i_1}\cdots A_{i_N}) \ket{i_1,\ldots,i_N},
 	\end{align}
 where the $\{A_i\}_{i=1}^d$ are a set of $D\times D$ matrices ($D$ in this case called \emph{bond dimension} and is related to, but not the same as the $D$ above). In the following we assume that the $A_i$ and the bond dimension are constant, i.e., independent of $N$.
 Generic MPS are so-called \emph{injective} MPS, which essentially means that the matrices $A_i$ cannot be brought into a common block-diagonal form by a similarity transformation. 
 In this case the MPS has a finite-correlation length. It is useful to think of injective MPS as generalizations from the above setting of local quantum circuits to the case of quasi-local quantum circuits.
 Indeed, for injective MPS it can be shown (using the results of Ref.~\cite{Wolf2008}) 
 that the reduced state on a sub-lattice $\tilde \Lambda$ in which the spins are separated by a distance $l$ is exponentially close to a product state,
 \begin{align}
 	\norm{\rho_{\tilde \Lambda}-\bigotimes_{x\in\tilde\Lambda}\rho_x}_1 \leq |\tilde\Lambda| C \e^{-l/\xi},
 \end{align}
 for some constants $C,\xi>0$. Thus, at least if $l$ grows weakly with the system size as $\log(N)$, this strongly suggests that the R\'enyi entropies are extensive on $\tilde \Lambda$. This in turn leads to a R\'enyi entanglement entropy of order $N/\log(N)$ and, correspondingly, to equilibration with a precision of order $\exp(-N/\log(N))$.  
 However, we have not been able to give a formal proof of this statement so far.
 
 Let us remark, however, that for the proof of Theorem~\ref{thm:main} it suffices to show that energy eigenstates with finite energy density have exponentially small overlap with product states.
 We will now prove that generic, translationally invariant MPS fulfill this criterion. Thus exponentially good equilibration also follows if the eigenstates of the Hamiltonian are generic MPS with sufficiently strong entanglement. While we assume translational invariance here, we expect that similar results also hold (with high probability) in the case where the local MPS-tensors are drawn at random. 
 
 \begin{lemma}[Exponentially small overlap of MPS with product states]
 	Consider translationally invariant, injective MPSs $\ket{\mathrm{MPS}[A]_N}$ with constant bond dimension.
 	Then either $\ket{\mathrm{MPS}[A]_N}$ is a product or there exists a constant $\kappa>0$
 	such that
 	\begin{align}
 		|\langle \phi|^{\otimes N} \ket{\mathrm{MPS}[A]_N}| \leq \e^{-\kappa N}
 	\end{align}
 	for all $\ket{\phi}\in\CC^d$.
 \end{lemma}
 Before giving the proof of this Lemma, let us note that we have chosen the product state here to be translationally invariant for convenience. One would arrive at the same conclusion if the product state in question has higher periodicity by coarse-graining the lattice. For example, it could be a charge-density wave of the form 
 $\ket{\uparrow ,\downarrow ,\uparrow ,\downarrow, \ldots}$. All that is required is that the state vector
 can be written as $\ket{\psi}^{N/m}$ for some $m$ corresponding to a coarse graining on the lattice. Then the MPS in question is either a product state on this coarse grained lattice or has exponentially small overlap with all product states on this coarse grained lattice. 
 \begin{proof}
 	The main ingredient for the proof is that an injective MPS is always the unique ground-state of a local, frustration free Hamiltonian \cite{Fannes1992,Nachtergaele1996,Perez-Garcia2006,Perez-Garcia2007}. 
 	This implies that for each lattice site $x$ there exists a local projector $P_x$ on a finite neighborhood $V_x$ of the lattice sites $x\in \Lambda$ such that
 	\begin{align}\label{eq:stabilizer}
 		P_x\otimes \one_{\Lambda\setminus V_x} \ket{\mathrm{MPS}[A]_N}= \ket{\mathrm{MPS}[A]_N}.
 	\end{align}
 	Furthermore, $\ket{\mathrm{MPS}[A]_N}$ is the only vector state fulfilling this condition for all $x$.
 	Let us choose a sub-lattice $\tilde \Lambda$ so that the neighborhoods $V_x$ do not overlap. Then the operators $P_x \otimes \one_{\Lambda\setminus V_x}$ commute. 
 	In the following, we will omit the identity for brevity of notation and simply write $P_x$ instead of $P_x \otimes \one_{\Lambda\setminus V_x}$. 
 	Let $\ket{\Phi} = \ket{\phi}^{\otimes N}$. Then we have
 	\begin{align}
 		\left|\langle\Phi|\mathrm{MPS}[A]_N\rangle\right|^2 &= \langle \Phi| \prod_{x\in\tilde \Lambda}P_x \proj{\mathrm{MPS}[A]_N} \prod_{y\in\tilde \Lambda}P_y |\Phi\rangle \\
 		&\leq \langle \Phi | \prod_{x\in\tilde \Lambda}P_x |\Phi \rangle \\
 		&= \prod_x \langle \Phi_{V_x}| P_x |\Phi_{V_x}\rangle = \e^{-c |\tilde\Lambda|},
 	\end{align}
 	where we have defined $\ket{\Phi_{V_x}} := \otimes_{x\in V_x} \ket{\phi}$ and implicitly defined $c>0$. 
 	In the first equality we used \eqref{eq:stabilizer}, in the second line we used that the projector onto the MPS is positive and in the last line we used that $\ket{\Phi}$ is a product state and that the $P_x$ are translations of each other since the MPS under question is translationally invariant.
 	Now suppose that $c=0$, i.e.,
 	\begin{align}
 	\langle \Phi_{V_x}| P_x |\Phi_{V_x}\rangle = 1. 
 	\end{align}
 	In this case we deduce that $P_x \ket{\Phi} = \ket{\Phi}$ for all $x$. However, by assumption $\ket{\mathrm{MPS}[A]_N}$ is the unique such state vector and we get $\ket{\mathrm{MPS}[A]_N}=\ket{\Phi}$.
 	Now consider
 	\begin{align}
 		\sup_{\ket{\phi}} \langle \Phi| P_X |\Phi\rangle \coloneqq \e^{-\tilde \kappa},
 	\end{align}
 	where again $\ket{\Phi} = \ket{\phi}^{\otimes N}$.
 	Since the the set of pure states is compact, the supremum is in fact achieved. But by assumption it cannot be equal to $1$, therefore $\tilde \kappa>0$. This step completes the proof by identifying
 	$\kappa>0$ as the constant that fulfills	$\kappa N = \tilde\kappa |\tilde \Lambda|$. 
 \end{proof}
 While we have above only considered the case of matrix-product states, we note that the proof applies to higher-dimensional tensor network states such as \emph{projected entangled pair states} (PEPS), 
 since the only property that we used is that the state in question is a unique, translationally invariant ground state of a local, frustration-free Hamiltonian.

\subsection{\ref*{app:prooftheorem}: Proof of Theorem~\ref{thm:main}}
Here we provide the proof of the main result. Before we spell out the details, let us explain the basic structure of the proof. 
Due to the monotonicity of the R\'enyi entropies in $\alpha$ it suffice to derive a lower bound for $S_\infty$ in order to get the statement of Theorem~\ref{thm:main} for all values of $\alpha \geq 0$. 
In general, there exists a basis of energy eigenstates that diagonalizes the time averaged state $\omega_\Lambda$. 
The eigenvalues of $\omega_\Lambda$ are given by the absolute-squared overlap of the initial state with these energy eigenstates. 
Therefore, an exponential upper bound on the maximal overlap of the initial product state vector $\ket{\Psi}_\Lambda$ with an energy eigenstate implies a lower bound on $S_\infty$. 
Let us denote the mean energy density of $\ket{\Psi}_\Lambda$ by $e$.
We establish a bound on the overlaps by two separate arguments -- one applying to eigenstates whose energy density deviates from $e$  by \emph{at most} $\delta$ and  one applying to eigenstates with energy density that deviates from $e$ by \emph{at least} $\delta$. 
Taken together we thus bound the overlaps with \emph{all} energy eigenstates for a given system size $|\Lambda|$, which establishes our final bound.
The constant $\delta$ is a small constant that we choose in the course of the proof. 

Before we continue, note that entanglement-ergodicity is defined for \emph{sequences} of Hamiltonians of increasing system sizes, whereas the statement of Theorem~\ref{thm:main} applies to every system size larger than $N_0(e)$, which will be identified below. 

As stated above, we first show that $\ket{\Psi}_\Lambda$ has little overlap with energy eigenstates that differ by a large amount in energy. This is due to the fact that the energy distribution of product states is sharply peaked around its mean. 
More precisely, we make use of the following Theorem due to Anshu \cite{Anshu2016}, whose formulation we slightly adapt to the current notation:
\begin{theorem}[Concentration bound \cite{Anshu2016}]
	Let $H_\Lambda$ be a strictly local, uniformly bounded Hamiltonian. Then there exists a constant $m>0$ independent of the system size such that for any product state  $\ket{\Psi}_\Lambda$with energy-density $e$ and any $\delta>\sqrt{2/(mN)}$, we have
	\begin{align}
		\bra{\Psi} \Pi_{(-\infty, N(e-\delta)]}\ket{\Psi} &\leq  \e^{-m \delta^2 N},\\
		\bra{\Psi} \Pi_{[N(e+\delta),\infty)}\ket{\Psi} &\leq  \e^{-m \delta^2 N},
	\end{align}
	where $\Pi_I$ denotes the projection onto the subspace spanned by energy eigenstates with energies in the interval $I$.
\end{theorem}

 In particular, this theorem implies that for any $\delta>\sqrt{2/(mN)}$, a product state vector $\ket{\Psi}_\Lambda$ has exponentially small overlap with energy eigenvectors  $\ket{e'}_\Lambda$ whose energy density $e'$ deviates by more than $\delta$ from its mean energy $e$,
 \begin{align}\label{eq:proof:anshu}
	 |\braket{\Psi}{e'}_\Lambda|^2 \leq \e^{-m \delta^2 N}\ \ \text{if}\  |e' - e| \geq \delta>\sqrt{2/(mN)} 
 \end{align}

 It remains to bound the overlap of $\ket{\Psi}_\Lambda$, which has energy density $e$, with eigenstates with energy density within the interval $[e-\delta,e+\delta]$. Here, the main ingredient is the fact that product states have exponentially small overlap with energy eigenstates that fulfill a weak volume law. 
 By the definition of entanglement-ergodic systems, for sufficiently large systems and for any energy eigenstate state $\rho_\Lambda(e') = \ketbra{e'}{e'}_\Lambda$ with energy density $e'$ there exists a region $A_\Lambda$ such that $S_{2}(\rho_{A_\Lambda}(e')) \geq g(e') N$ provided that $N\geq N_0$ (where $N_0$ is given by entanglement-ergodicity). 
 Using Lemma~\ref{lemma:overlap}, we find that the overlap of a product state with eigenstates of energy density $e'$ is hence bounded as
 \begin{align}
 	|\braket{\Psi}{e'}_\Lambda|^2 \leq \e^{-\frac{1}{2} g(e') N}. 
 \end{align}
 We now use the regularity of $g$. 
 Let $K>0$ be its Lip\-schitz-constant. 
 Then we know that $g(e') \geq g(e)-K\delta$ for all $e'\in[e-\delta,e+\delta]$. 
 We choose
 \begin{align}
	 N_0(e) := \max\left\{N_0, \frac{8 K^2}{m g(e)^2}\right\},\quad \delta = \frac{g(e)}{2K}.
\end{align}
This ensures that $\delta > \sqrt{2/(mN)}$ for all $N> N_0(e)$. 
Then we obtain
 \begin{align}\label{eq:proof:ergodicity}
 	|\braket{\psi}{e'}_\lambda|^2 \leq \e^{-\frac{1}{4} g(e) N},\quad e'\in[e-\delta,e+\delta]. 
 \end{align}
 In summary, we find that \emph{all} eigenstates have exponentially small overlap with $\ket{\Psi}_\Lambda$, either because their energy density differs substantially from $e$ or because they fulfill a weak volume law. 
 Taking 
 \begin{align}
 	k(e) = \frac{1}{4}g(e) \operatorname{min}\{1, m\,g(e)/K^2 \}, 
 \end{align}
 we thus find $|\braket{\psi}{e_i}_\lambda|^2 \leq \exp(-k(e) N)$
for all energy eigenstates $\ket{e_i}_\Lambda$.
Therefore
\begin{align}
	S_\alpha(\omega_\Lambda)\geq S_\infty(\omega_\Lambda) \geq k(e) N,
\end{align}
which completes the proof.

 \subsection{\ref*{app:quasilocal}: Stability of entanglement-ergodicity under quasi-local unitaries}
 In this section we prove the stability of entanglement-ergodicity under time evolution generated by a local Hamiltonian.
 More precisely, we establish the following result:
 \begin{proposition}[Stability under quasi-local unitaries]\label{lemma:stability}
 Let $H_\Lambda$ be an entanglement-ergodic system where the regions $A_\Lambda$ have asymptotically vanishing surface-to-volume ratio and let $U_\Lambda'$ be a quasi-local sequence of unitaries. 
 Then the system with Hamiltonian $U_\Lambda' H_\Lambda U_\Lambda'^\dagger$ is entanglement-ergodic. 
 \end{proposition}
 Here, by a quasi-local unitary, we mean the following: 
 Consider a sequence of unitaries $U_\Lambda'$. 
 We call the sequence $U_\Lambda'$ quasi-local if there exists a sequence of $l$-local Hamiltonians $H'_\Lambda=\sum_x h'_x$ with $\norm{h'_x}\leq 1$ and a finite time $T$ such that
 \begin{align}
 U_\Lambda' = \e^{-\ii H'_\Lambda t_\Lambda},\quad t_\Lambda\leq T.
 \end{align}
 Examples of quasi-local unitaries are given by finite-depth local quantum circuits, in particular ones that reflect
 deformations of a quantum state within the same quantum phase \cite{PhysRevB.82.155138}.
 This result above follows from the following Lemma. It can be seen as the R\'enyi-2 analogue of the result on entanglement
 rates for the von Neumann entropy laid out in Ref.\ \cite{PhysRevLett.111.170501}.
 
 \begin{lemma}[R\'enyi-2 entangling rate]\label{lemma:entangling}
 	Let $\rho_{AB}$ be a bipartite state that evolves under a Hamiltonian $H=H_A + H_B + V$ with $H_{A,B}$ acting on systems $A,B$, respectively. Decompose $V$ into a Hermitian operator basis as $V= \sum_{i,j=1}^v c_{i,j} A_i\otimes B_j$ with $\norm{A_i}=\norm{B_j}=1$ and $A_i^\dagger= A_i, B_j^\dagger=B_j$ for all $i,j$.
 We define the R\'enyi-2 entangling rate of the subsystem $A$ as 
 \begin{align}
 	\Gamma(A,H) \coloneqq	\left.\frac{\dd S_2(\rho_A)}{\dd t}\right|_{t=0}= 2\ii \frac{\tr(\rho_A \tr_B([V,\rho_{AB}]))}{\tr(\rho_A^2)}.
 \end{align}
 Then 
 \begin{align}
 	|\Gamma(A,H)| \leq 4\norm{\CC}_1 = 4 \sum_{i,j}|c_{i,j}|.
 \end{align}
 \end{lemma}
 Before providing the proof of the Lemma, let us discuss the application to a strictly
 local many-body systems. In this case $H=\sum_x h_x$ with $\norm{h_x}\leq 1$
 and where the support of each term $h_x$ has diameter bounded by some constant. 
 Choose a region $A$ and decompose $H=H_A +
 H_{A^c} + H_{\partial A}$, where $H_{\partial A}$ contains only those terms
 $h_x$ that have support both on $A$ and $A^c$. Clearly, the number of these terms
 scales like the boundary of $A$. We then have \begin{align}
 	\Gamma(A,H) &= 2\ii \frac{\tr\left(\rho_A \tr_{A^c}([H,\rho])\right)}{\Tr\left(\rho_A^2\right)} =  2\ii \left[\frac{\tr\left(\rho_A \tr_{A^c}([H_{\partial A},\rho])\right)}{\Tr\left(\rho_A^2\right)}+ \frac{\tr\left(\rho_A \tr_{A^c}([H_{A},\rho])\right)}{\Tr\left(\rho_A^2\right)}+ \frac{\tr\left(\rho_A \tr_{A^c}([H_{A^c},\rho])\right)}{\Tr\left(\rho_A^2\right)} \right]\\
 				  &=2\ii \left[\frac{\tr\left(\rho_A \tr_{A^c}([H_{\partial A},\rho])\right)}{\Tr\left(\rho_A^2\right)}+ \frac{\tr\left(\rho_A [H_{A},\rho_A]\right)}{\Tr\left(\rho_A^2\right)}+ \frac{\tr\left(\rho_A \tr_{A^c}([H_{A^c},\rho])\right)}{\Tr\left(\rho_A^2\right)} \right]\\
 	     &= 2\ii \frac{\tr\left(\rho_A \tr_{A^c}([H_{\partial A},\rho])\right)}{\Tr\left(\rho_A^2\right)} 
 	= \sum_{x\in\partial A} 2\ii \frac{\tr\left(\rho_A \tr_{A^c}([h_{x},\rho])\right)}{\Tr\left(\rho_A^2\right)} 
 	=\sum_{x\in\partial A}\Gamma(A,h_x),
 \end{align}
 where $x\in\partial A$ denotes those $x$ for which the support of $h_x$ is neither fully contained in $A$ nor in $A^c$. 
 The third line follows from cyclicity of the trace for the second term and by choosing the eigenbasis of $H_{A^c}$ to trace-out system $A^c$ for the third term in the second line.
 We can now use the lemma to bound $|\Gamma(A,h_x)| \leq 4\norm{\mathbf{C}_x}_1$, where 
 $\mathbf{C}$ is the vector of coefficients $c_{i,j}$ associated to the Hamiltonian term $h_x$. We hence obtain
 \begin{align}
 	|\Gamma(A,H)| \leq 4\,|\partial A|\, \max_{x\in\Lambda}\, \norm{\mathbf{C}_x}_1.
 \end{align}
 Since we assume that the Hamiltonian is strictly local and uniformly bounded, there is a maximum value of $\norm{\mathbf{C}_x}_1$ independent of the system size. 
 Integrating over time we obtain
 \begin{align}
 	\left|S_2(\rho_A(t)) - S_2(\rho_A(0))\right| \leq 4 t\, |\partial A|\, \max_{x\in\Lambda}\norm{\mathbf{C}_x}_1.
 \end{align}
 Thus, time evolution under a local Hamiltonian for a fixed time $t$ (independent of the system size) cannot change the volume law of the R\'enyi-2 entropy as long as the surface area of $A$ is asymptotically vanishingly small in comparison with its volume. Using the equivalence of volume laws for different values of $\alpha>1$, we thus find that an entanglement-ergodic system remains entanglement-ergodic (possibly with a renormalized function $g$) under such time evolution.  
 Clearly, the same is true if a finite number of such time evolutions $U_\Lambda'$ is concatenated as long as this number does not grow with the size of the lattice.

 \begin{proof}[Proof of Lemma~\ref{lemma:entangling}]
 	First we write the entangling rate as (using the same method as above)
 \begin{align}
 	\Gamma(A,H) = 2\ii \frac{\tr(\rho_A \tr_B([V,\rho_{AB}]))}{\tr(\rho_A^2)}.
 \end{align}
 We now use $\tr_B([A\otimes B,\rho_{AB}]) = [A,\tr_B(\rho_{AB}B)]$, which follows from the cyclicity of the trace. This allows us to write
 \begin{align}
 	|\tr(\rho_A\tr_B([V,\rho_{AB}]))| &\leq \sum_{i,j} |c_{i,j}| |\tr(\rho_A [A_i,\tr_B(\rho_{AB}B_j)])|\\
 	 &\leq 2 \sum_{i,j} |c_{i,j}| \left[\tr(\rho_A^2 A^2_i)\tr(\tr_B(\rho_{AB}B_j)^\dagger \tr_B(\rho_{AB}B_j))\right]^{1/2},
 \end{align}
 where we have used the Cauchy-Schwarz inequality. We can now use the fact that $\norm{A_i}=1$ to obtain 
 \begin{align}
 	|\tr(\rho_A\tr_B([V,\rho_{AB}]))| &\leq 2 \sum_{i,j} |c_{i,j}| \left[\tr(\rho_A^2)\tr(\tr_B(\rho_{AB}B_j)^\dagger\tr_B(\rho_{AB}B_j))\right]^{1/2}.
 \end{align}
 Let us now consider the term $\tr\left(\tr_B(\rho_{AB}B_j)^\dagger \tr_B(\rho_{AB}B_j)\right)$. Let us decompose $B_j$ into its eigenbasis and write
 \begin{align}
 	\tr_B(\rho_{AB}B_j) = \sum_\alpha b^\alpha_j \tr_B(\rho_{AB}P^\alpha_j),
 \end{align}
 where $P^\alpha_j$ are the eigenprojectors of $B_j$ and $b^\alpha_j$ its eigenvalues.
 Upper bounding each $b^\alpha_j$ by the maximum one, we then find the operator inequality
 \begin{align}
 	\tr_B(\rho_{AB}B_j)^\dagger \tr_B(\rho_{AB}B_j) &= \sum_{\alpha,\beta} b^\alpha_j b^\beta_j \tr_B(\rho_{AB}P^\alpha_j)\tr_B(\rho_{AB}P^\beta_j)\\
 						&\leq \norm{B_{j}}^2 \tr_B(\rho_{AB})^\dagger \tr_B(\rho_{AB})\\
 		&= \norm{B_{j}}^2 \rho_A^2.
 \end{align}
 Since $\norm{B_j}=1$, we thus finally obtain
 \begin{align}
 	|\tr(\rho_A\tr_B([V,\rho_{AB}]))| &\leq 2 \sum_{i,j} |c_{i,j}| \tr(\rho_A^2) = 2 \norm{\mathbf{C}}_1 \tr(\rho_A^2).
 \end{align}
 We thus obtain for the entangling rate
 \begin{align}
 	|\Gamma(A,H)| \leq 4 \norm{\mathbf{C}}_1.
 \end{align}
 This step completes the proof. 
 \end{proof}
 
 The above lemma shows stability of entanglement-ergodic systems under arbitrary quasi-local unitaries. 
 We would like to highlight that stability under finite-depth quantum circuits follows immediately from the fact that finite-depth quantum circuits can only change 
 the R\'enyi entropy of some region by an amount proportional to the area of the boundary of the region. This is due to the fact that a finite-depth quantum circuit can only (de-)correlate degrees of freedom that 
 are within a finite distance from each other.
 
 \subsection{\ref*{app:state_quasilocal}: Quasi-local initial states}
 Here, we discuss the case of initial state vectors $\ket{\phi_\Lambda}=U_\Lambda'\ket{\Psi}_\Lambda$, where $U_\Lambda'$ is an arbitrary quasi-local unitary and $\ket{\Psi}_\Lambda$ is a product vector.
 In this case, we can show that
 \begin{align}
 	S_\alpha(\omega) \geq k(e) N^{\frac{1}{\nu +1}},
 \end{align}
 which in turn implies equilibration with precision of order $\exp(-N^{\frac{1}{\nu+1}})$.
 To see this, first note that due to stability of entanglement-ergodicity under quasi-local unitaries almost all the steps of the proof of Theorem~\ref{thm:main} transfer directly. 
 In particular the overlap of $\ket{\Psi}_\Lambda$ with energy eigenstates of energy density $e'\in [e-\delta,e+\delta]$ is exponentially small in $N$ due to entanglement-ergodicity. 
 The only step which cannot be transferred is the argument that the overlap with energy eigenstates whose energy density differs by more than $\delta$ from $e$ is exponentially small. 
 Indeed, the best general result for quasi-local states known to us, from Ref.~\cite{Anshu2016}, shows that this overlap is only sub-exponentially small. 
 More precisely, the following theorem holds:
 \begin{theorem}[\cite{Anshu2016}]
 Let $H$ be a strictly local Hamiltonian on a $\nu$-dimensional lattice with $N$ sites, let $\rho$ be a quantum state with correlation length $\xi>0$ and $\langle H\rangle_\rho=  {\rm tr}(H \rho)$ be the average energy of $\rho$.
 Furthermore denote by  $\Pi_{X}$ the projector onto all energy eigenstates with energy in the set $X\subset \mathbb{R}$. For $a\geq (2^{O(\nu)}/N \xi )^{1/2}$, it holds that
 \begin{eqnarray}
 	{\rm tr} (\rho \Pi_{[\langle H\rangle_\rho+Na,\infty]})
 	\leq O(\xi)
 	\exp\left(
 	-\frac{(N a^2 \xi)^{1/(\nu+1)}}{O(1) \nu \xi}
 	\right)
 \end{eqnarray}
 and
 \begin{eqnarray}
 	{\rm tr} (\rho \Pi_{[0, \langle H\rangle_\rho-Na]} ) \leq O(\xi) \exp\left(
 	-\frac{(N a^2 \xi)^{1/(\nu+1)}}{O(1) \nu \xi}
 	\right).
 	\end{eqnarray}
 \end{theorem}
 Note that this theorem not only bounds the overlap with individual energy eigenstates, as is needed to prove our result, but in fact the
 cumulative overlap with \emph{all} energy eigenstates whose energy density deviates from $e$. 
 Since these are the vast majority of all states in the exponentially large
 Hilbert-space, it seems highly plausible that despite the fact that the
 cumulative overlap with all these states only scales sub-exponentially, the overlap with each individual such energy eigenstate is exponentially small, which would imply exponentially good equilibration also for quasi-local initial states.  
 We thus conjecture that (possibly under mild additional assumptions) it is possible to prove exponentially good equilibration in entanglement-ergodic systems for arbitrary initial states with a finite correlation length. 
 
 \subsection{\ref*{app:eth}: Entanglement ergodicity and eigenstate thermalization}
 In this section we argue that entanglement ergodicity follows from the eigenstate thermalization hypothesis, provided that one assumes that the latter applies to subsystems that contain a finite fraction of the total system. While we believe that our observations make this statement highly plausible, we leave the formal proof as an open problem.
 We thus assume that for any energy density $e$, the reduced state of some subsystem $A$ is given by the reduced state of the Gibbs state with inverse temperature $\beta(e)$ corresponding to the energy density $e$,
 \begin{align}
 	\tr_{A^c}\left(\proj{e}\right) = \tr_{A^c}\left(\frac{\e^{-\beta(e) H}}{Z_{\beta(e)}}\right).
 \end{align}
 We further only consider energy densities for which the thermal state has a finite correlation length and for simplicity assume that the system is translationally invariant.
 In this case the thermal state of the total system at inverse temperature $\beta$ has a finite equilibrium free energy density $f(\beta)$. 
 Let us take the convention that the ground-state of the Hamiltonian has energy zero. In this convention, the total free energy $F_\beta=Nf(\beta)=-\frac{1}{\beta}\log(Z_\beta)$ is negative. 
 Furthermore the ground-state probability $p(\beta)$ is given by 
 \begin{align}
 	p(\beta) = 1/Z = \e^{\beta F_\beta}
 \end{align}
 and hence the R\'enyi-$\infty$ entropy of the total system is given by
 \begin{align}
 S_\infty = - \beta F_\beta,
 \end{align}
 which is a positive, extensive quantity, since the free energy is extensive. 
 We thus correspondingly expect that to leading order, i.e., up to boundary terms, the corresponding entropy of a contiguous subsystem $A$ is given by 
 \begin{align}
 	S_\infty(A) = -\beta|A|f(\beta(e)),
 \end{align}
 which then implies entanglement-ergodicity. This is further supported by the fact that the reduced state on region $A$ is well approximated by the reduced state of 
 the Gibbs state of the Hamiltonian $H$ restricted to $A$ plus a "buffer" region around it -- a phenomenon known as locality of temperature \cite{Kliesch2014}.


\begin{thebibliography}{75}%
\makeatletter
\providecommand \@ifxundefined [1]{%
 \@ifx{#1\undefined}
}%
\providecommand \@ifnum [1]{%
 \ifnum #1\expandafter \@firstoftwo
 \else \expandafter \@secondoftwo
 \fi
}%
\providecommand \@ifx [1]{%
 \ifx #1\expandafter \@firstoftwo
 \else \expandafter \@secondoftwo
 \fi
}%
\providecommand \natexlab [1]{#1}%
\providecommand \enquote  [1]{``#1''}%
\providecommand \bibnamefont  [1]{#1}%
\providecommand \bibfnamefont [1]{#1}%
\providecommand \citenamefont [1]{#1}%
\providecommand \href@noop [0]{\@secondoftwo}%
\providecommand \href [0]{\begingroup \@sanitize@url \@href}%
\providecommand \@href[1]{\@@startlink{#1}\@@href}%
\providecommand \@@href[1]{\endgroup#1\@@endlink}%
\providecommand \@sanitize@url [0]{\catcode `\\12\catcode `\$12\catcode
  `\&12\catcode `\#12\catcode `\^12\catcode `\_12\catcode `\%12\relax}%
\providecommand \@@startlink[1]{}%
\providecommand \@@endlink[0]{}%
\providecommand \url  [0]{\begingroup\@sanitize@url \@url }%
\providecommand \@url [1]{\endgroup\@href {#1}{\urlprefix }}%
\providecommand \urlprefix  [0]{URL }%
\providecommand \Eprint [0]{\href }%
\providecommand \doibase [0]{http://dx.doi.org/}%
\providecommand \selectlanguage [0]{\@gobble}%
\providecommand \bibinfo  [0]{\@secondoftwo}%
\providecommand \bibfield  [0]{\@secondoftwo}%
\providecommand \translation [1]{[#1]}%
\providecommand \BibitemOpen [0]{}%
\providecommand \bibitemStop [0]{}%
\providecommand \bibitemNoStop [0]{.\EOS\space}%
\providecommand \EOS [0]{\spacefactor3000\relax}%
\providecommand \BibitemShut  [1]{\csname bibitem#1\endcsname}%
\let\auto@bib@innerbib\@empty
\bibitem [{\citenamefont {Tasaki}(1998)}]{Tasaki1998}%
  \BibitemOpen
  \bibfield  {author} {\bibinfo {author} {\bibfnamefont {H.}~\bibnamefont
  {Tasaki}},\ }\href {\doibase 10.1103/PhysRevLett.80.1373} {\bibfield
  {journal} {\bibinfo  {journal} {Phys. Rev. Lett.}\ }\textbf {\bibinfo
  {volume} {80}},\ \bibinfo {pages} {1373} (\bibinfo {year}
  {1998})}\BibitemShut {NoStop}%
\bibitem [{\citenamefont {Reimann}(2008)}]{Reimann2008}%
  \BibitemOpen
  \bibfield  {author} {\bibinfo {author} {\bibfnamefont {P.}~\bibnamefont
  {Reimann}},\ }\href {\doibase 10.1103/PhysRevLett.101.190403} {\bibfield
  {journal} {\bibinfo  {journal} {Phys. Rev. Lett.}\ }\textbf {\bibinfo
  {volume} {101}},\ \bibinfo {pages} {190403} (\bibinfo {year} {2008})},\
  \Eprint {http://arxiv.org/abs/0810.3092} {arXiv:0810.3092} \BibitemShut
  {NoStop}%
\bibitem [{\citenamefont {Linden}\ \emph {et~al.}(2009)\citenamefont {Linden},
  \citenamefont {Popescu}, \citenamefont {Short},\ and\ \citenamefont
  {Winter}}]{Linden2009}%
  \BibitemOpen
  \bibfield  {author} {\bibinfo {author} {\bibfnamefont {N.}~\bibnamefont
  {Linden}}, \bibinfo {author} {\bibfnamefont {S.}~\bibnamefont {Popescu}},
  \bibinfo {author} {\bibfnamefont {A.}~\bibnamefont {Short}}, \ and\ \bibinfo
  {author} {\bibfnamefont {A.}~\bibnamefont {Winter}},\ }\href {\doibase
  10.1103/PhysRevE.79.061103} {\bibfield  {journal} {\bibinfo  {journal} {Phys.
  Rev. E}\ }\textbf {\bibinfo {volume} {79}},\ \bibinfo {pages} {61103}
  (\bibinfo {year} {2009})},\ \Eprint {http://arxiv.org/abs/0812.2385}
  {arXiv:0812.2385} \BibitemShut {NoStop}%
\bibitem [{\citenamefont {Goldstein}\ \emph {et~al.}(2010)\citenamefont
  {Goldstein}, \citenamefont {Lebowitz}, \citenamefont {Mastrodonato},
  \citenamefont {Tumulka},\ and\ \citenamefont {Zangh{\`\i}}}]{Goldstein2010}%
  \BibitemOpen
  \bibfield  {author} {\bibinfo {author} {\bibfnamefont {S.}~\bibnamefont
  {Goldstein}}, \bibinfo {author} {\bibfnamefont {J.~L.}\ \bibnamefont
  {Lebowitz}}, \bibinfo {author} {\bibfnamefont {C.}~\bibnamefont
  {Mastrodonato}}, \bibinfo {author} {\bibfnamefont {R.}~\bibnamefont
  {Tumulka}}, \ and\ \bibinfo {author} {\bibfnamefont {N.}~\bibnamefont
  {Zangh{\`\i}}},\ }\href {\doibase 10.1098/rspa.2009.0635} {\bibfield
  {journal} {\bibinfo  {journal} {Proc. Roy. Soc. A}\ }\textbf {\bibinfo
  {volume} {466}},\ \bibinfo {pages} {3203} (\bibinfo {year}
  {2010})}\BibitemShut {NoStop}%
\bibitem [{\citenamefont {Short}(2011)}]{Short2011}%
  \BibitemOpen
  \bibfield  {author} {\bibinfo {author} {\bibfnamefont {A.~J.}\ \bibnamefont
  {Short}},\ }\href {\doibase 10.1088/1367-2630/13/5/053009} {\bibfield
  {journal} {\bibinfo  {journal} {New J. Phys.}\ }\textbf {\bibinfo {volume}
  {13}},\ \bibinfo {pages} {053009} (\bibinfo {year} {2011})}\BibitemShut
  {NoStop}%
\bibitem [{\citenamefont {Reimann}(2012)}]{Reimann2012}%
  \BibitemOpen
  \bibfield  {author} {\bibinfo {author} {\bibfnamefont {P.}~\bibnamefont
  {Reimann}},\ }\href {\doibase 10.1088/0031-8949/86/05/058512} {\bibfield
  {journal} {\bibinfo  {journal} {Phys. Scr.}\ }\textbf {\bibinfo {volume}
  {86}},\ \bibinfo {pages} {58512} (\bibinfo {year} {2012})},\ \Eprint
  {http://arxiv.org/abs/1210.5821} {arXiv:1210.5821} \BibitemShut {NoStop}%
\bibitem [{\citenamefont {Reimann}\ and\ \citenamefont
  {Kastner}(2012)}]{Reimann2012a}%
  \BibitemOpen
  \bibfield  {author} {\bibinfo {author} {\bibfnamefont {P.}~\bibnamefont
  {Reimann}}\ and\ \bibinfo {author} {\bibfnamefont {M.}~\bibnamefont
  {Kastner}},\ }\href {\doibase 10.1088/1367-2630/14/4/043020} {\bibfield
  {journal} {\bibinfo  {journal} {New J. Phys.}\ }\textbf {\bibinfo {volume}
  {14}},\ \bibinfo {pages} {43020} (\bibinfo {year} {2012})},\ \Eprint
  {http://arxiv.org/abs/1202.2768} {arXiv:1202.2768} \BibitemShut {NoStop}%
\bibitem [{\citenamefont {Short}\ and\ \citenamefont
  {Farrelly}(2012)}]{Short2012}%
  \BibitemOpen
  \bibfield  {author} {\bibinfo {author} {\bibfnamefont {A.~J.}\ \bibnamefont
  {Short}}\ and\ \bibinfo {author} {\bibfnamefont {T.~C.}\ \bibnamefont
  {Farrelly}},\ }\href {\doibase 10.1088/1367-2630/14/1/013063} {\bibfield
  {journal} {\bibinfo  {journal} {New J. Phys.}\ }\textbf {\bibinfo {volume}
  {14}},\ \bibinfo {pages} {013063} (\bibinfo {year} {2012})}\BibitemShut
  {NoStop}%
\bibitem [{\citenamefont {Masanes}\ \emph {et~al.}(2013)\citenamefont
  {Masanes}, \citenamefont {Roncaglia},\ and\ \citenamefont
  {Ac{\'{i}}n}}]{Masanes2013}%
  \BibitemOpen
  \bibfield  {author} {\bibinfo {author} {\bibfnamefont {L.}~\bibnamefont
  {Masanes}}, \bibinfo {author} {\bibfnamefont {A.~J.}\ \bibnamefont
  {Roncaglia}}, \ and\ \bibinfo {author} {\bibfnamefont {A.}~\bibnamefont
  {Ac{\'{i}}n}},\ }\href {\doibase 10.1103/PhysRevE.87.032137} {\bibfield
  {journal} {\bibinfo  {journal} {Phys. Rev. E}\ }\textbf {\bibinfo {volume}
  {87}},\ \bibinfo {pages} {32137} (\bibinfo {year} {2013})},\ \Eprint
  {http://arxiv.org/abs/1108.0374} {arXiv:1108.0374} \BibitemShut {NoStop}%
\bibitem [{\citenamefont {Gogolin}\ and\ \citenamefont
  {Eisert}(2016)}]{Gogolin2016}%
  \BibitemOpen
  \bibfield  {author} {\bibinfo {author} {\bibfnamefont {C.}~\bibnamefont
  {Gogolin}}\ and\ \bibinfo {author} {\bibfnamefont {J.}~\bibnamefont
  {Eisert}},\ }\href {\doibase 10.1088/0034-4885/79/5/056001} {\bibfield
  {journal} {\bibinfo  {journal} {Rep. Prog. Phys.}\ }\textbf {\bibinfo
  {volume} {79}},\ \bibinfo {pages} {56001} (\bibinfo {year} {2016})},\ \Eprint
  {http://arxiv.org/abs/1503.07538} {arXiv:1503.07538} \BibitemShut {NoStop}%
\bibitem [{\citenamefont {{Eisert}}\ \emph {et~al.}(2015)\citenamefont
  {{Eisert}}, \citenamefont {{Friesdorf}},\ and\ \citenamefont
  {{Gogolin}}}]{EisFriGog15}%
  \BibitemOpen
  \bibfield  {author} {\bibinfo {author} {\bibfnamefont {J.}~\bibnamefont
  {{Eisert}}}, \bibinfo {author} {\bibfnamefont {M.}~\bibnamefont
  {{Friesdorf}}}, \ and\ \bibinfo {author} {\bibfnamefont {C.}~\bibnamefont
  {{Gogolin}}},\ }\href {\doibase 10.1038/nphys3215} {\bibfield  {journal}
  {\bibinfo  {journal} {Nature Phys.}\ }\textbf {\bibinfo {volume} {11}},\
  \bibinfo {pages} {124} (\bibinfo {year} {2015})},\ \Eprint
  {http://arxiv.org/abs/1408.5148} {arXiv:1408.5148} \BibitemShut {NoStop}%
\bibitem [{\citenamefont {Balz}\ and\ \citenamefont
  {Reimann}(2017)}]{PhysRevLett.118.190601}%
  \BibitemOpen
  \bibfield  {author} {\bibinfo {author} {\bibfnamefont {B.~N.}\ \bibnamefont
  {Balz}}\ and\ \bibinfo {author} {\bibfnamefont {P.}~\bibnamefont {Reimann}},\
  }\href {\doibase 10.1103/PhysRevLett.118.190601} {\bibfield  {journal}
  {\bibinfo  {journal} {Phys. Rev. Lett.}\ }\textbf {\bibinfo {volume} {118}},\
  \bibinfo {pages} {190601} (\bibinfo {year} {2017})}\BibitemShut {NoStop}%
\bibitem [{\citenamefont {Bloch}\ \emph {et~al.}(2008)\citenamefont {Bloch},
  \citenamefont {Dalibard},\ and\ \citenamefont {Zwerger}}]{Bloch2008}%
  \BibitemOpen
  \bibfield  {author} {\bibinfo {author} {\bibfnamefont {I.}~\bibnamefont
  {Bloch}}, \bibinfo {author} {\bibfnamefont {J.}~\bibnamefont {Dalibard}}, \
  and\ \bibinfo {author} {\bibfnamefont {W.}~\bibnamefont {Zwerger}},\ }\href
  {\doibase 10.1103/RevModPhys.80.885} {\bibfield  {journal} {\bibinfo
  {journal} {Rev. Mod. Phys.}\ }\textbf {\bibinfo {volume} {80}},\ \bibinfo
  {pages} {885} (\bibinfo {year} {2008})}\BibitemShut {NoStop}%
\bibitem [{\citenamefont {Bloch}\ \emph {et~al.}(2012)\citenamefont {Bloch},
  \citenamefont {Dalibard},\ and\ \citenamefont {Nascimbene}}]{Bloch2012}%
  \BibitemOpen
  \bibfield  {author} {\bibinfo {author} {\bibfnamefont {I.}~\bibnamefont
  {Bloch}}, \bibinfo {author} {\bibfnamefont {J.}~\bibnamefont {Dalibard}}, \
  and\ \bibinfo {author} {\bibfnamefont {S.}~\bibnamefont {Nascimbene}},\
  }\href {\doibase 10.1038/nphys2259} {\bibfield  {journal} {\bibinfo
  {journal} {Nature Phys.}\ }\textbf {\bibinfo {volume} {8}},\ \bibinfo {pages}
  {267} (\bibinfo {year} {2012})}\BibitemShut {NoStop}%
\bibitem [{\citenamefont {Schneider}\ \emph {et~al.}(2012)\citenamefont
  {Schneider}, \citenamefont {Hackerm{\"u}ller}, \citenamefont {Ronzheimer},
  \citenamefont {Will}, \citenamefont {Braun}, \citenamefont {Best},
  \citenamefont {Bloch}, \citenamefont {Demler}, \citenamefont {Mandt},
  \citenamefont {Rasch},\ and\ \citenamefont {Rosch}}]{Schneider2012}%
  \BibitemOpen
  \bibfield  {author} {\bibinfo {author} {\bibfnamefont {U.}~\bibnamefont
  {Schneider}}, \bibinfo {author} {\bibfnamefont {L.}~\bibnamefont
  {Hackerm{\"u}ller}}, \bibinfo {author} {\bibfnamefont {J.~P.}\ \bibnamefont
  {Ronzheimer}}, \bibinfo {author} {\bibfnamefont {S.}~\bibnamefont {Will}},
  \bibinfo {author} {\bibfnamefont {S.}~\bibnamefont {Braun}}, \bibinfo
  {author} {\bibfnamefont {T.}~\bibnamefont {Best}}, \bibinfo {author}
  {\bibfnamefont {I.}~\bibnamefont {Bloch}}, \bibinfo {author} {\bibfnamefont
  {E.}~\bibnamefont {Demler}}, \bibinfo {author} {\bibfnamefont
  {S.}~\bibnamefont {Mandt}}, \bibinfo {author} {\bibfnamefont
  {D.}~\bibnamefont {Rasch}}, \ and\ \bibinfo {author} {\bibfnamefont
  {A.}~\bibnamefont {Rosch}},\ }\href {\doibase 10.1038/nphys2205} {\bibfield
  {journal} {\bibinfo  {journal} {Nature Phys.}\ }\textbf {\bibinfo {volume}
  {8}},\ \bibinfo {pages} {213} (\bibinfo {year} {2012})}\BibitemShut {NoStop}%
\bibitem [{\citenamefont {Trotzky}\ \emph {et~al.}(2012)\citenamefont
  {Trotzky}, \citenamefont {Chen}, \citenamefont {Flesch}, \citenamefont
  {McCulloch}, \citenamefont {Schollw\"ock}, \citenamefont {Eisert},\ and\
  \citenamefont {Bloch}}]{Trotzky2012}%
  \BibitemOpen
  \bibfield  {author} {\bibinfo {author} {\bibfnamefont {S.}~\bibnamefont
  {Trotzky}}, \bibinfo {author} {\bibfnamefont {Y.-A.}\ \bibnamefont {Chen}},
  \bibinfo {author} {\bibfnamefont {A.}~\bibnamefont {Flesch}}, \bibinfo
  {author} {\bibfnamefont {I.~P.}\ \bibnamefont {McCulloch}}, \bibinfo {author}
  {\bibfnamefont {U.}~\bibnamefont {Schollw\"ock}}, \bibinfo {author}
  {\bibfnamefont {J.}~\bibnamefont {Eisert}}, \ and\ \bibinfo {author}
  {\bibfnamefont {I.}~\bibnamefont {Bloch}},\ }\href {\doibase
  10.1038/nphys2232} {\bibfield  {journal} {\bibinfo  {journal} {Nature Phys.}\
  }\textbf {\bibinfo {volume} {8}},\ \bibinfo {pages} {325} (\bibinfo {year}
  {2012})}\BibitemShut {NoStop}%
\bibitem [{\citenamefont {Braun}\ \emph {et~al.}(2015)\citenamefont {Braun},
  \citenamefont {Friesdorf}, \citenamefont {Hodgman}, \citenamefont
  {Schreiber}, \citenamefont {Ronzheimer}, \citenamefont {Riera}, \citenamefont
  {del Rey}, \citenamefont {Bloch}, \citenamefont {Eisert},\ and\ \citenamefont
  {Schneider}}]{Braun2015}%
  \BibitemOpen
  \bibfield  {author} {\bibinfo {author} {\bibfnamefont {S.}~\bibnamefont
  {Braun}}, \bibinfo {author} {\bibfnamefont {M.}~\bibnamefont {Friesdorf}},
  \bibinfo {author} {\bibfnamefont {S.~S.}\ \bibnamefont {Hodgman}}, \bibinfo
  {author} {\bibfnamefont {M.}~\bibnamefont {Schreiber}}, \bibinfo {author}
  {\bibfnamefont {J.~P.~p.}\ \bibnamefont {Ronzheimer}}, \bibinfo {author}
  {\bibfnamefont {A.}~\bibnamefont {Riera}}, \bibinfo {author} {\bibfnamefont
  {M.}~\bibnamefont {del Rey}}, \bibinfo {author} {\bibfnamefont
  {I.}~\bibnamefont {Bloch}}, \bibinfo {author} {\bibfnamefont
  {J.}~\bibnamefont {Eisert}}, \ and\ \bibinfo {author} {\bibfnamefont
  {U.}~\bibnamefont {Schneider}},\ }\href {\doibase 10.1073/pnas.1408861112}
  {\bibfield  {journal} {\bibinfo  {journal} {Proc. Natl. Ac. Sc.}\ }\textbf
  {\bibinfo {volume} {112}},\ \bibinfo {pages} {3641} (\bibinfo {year}
  {2015})}\BibitemShut {NoStop}%
\bibitem [{\citenamefont {Schreiber}\ \emph {et~al.}(2015)\citenamefont
  {Schreiber}, \citenamefont {Hodgman}, \citenamefont {Bordia}, \citenamefont
  {L{\"u}schen}, \citenamefont {Fischer}, \citenamefont {Vosk}, \citenamefont
  {Altman}, \citenamefont {Schneider},\ and\ \citenamefont
  {Bloch}}]{Schreiber2015}%
  \BibitemOpen
  \bibfield  {author} {\bibinfo {author} {\bibfnamefont {M.}~\bibnamefont
  {Schreiber}}, \bibinfo {author} {\bibfnamefont {S.~S.}\ \bibnamefont
  {Hodgman}}, \bibinfo {author} {\bibfnamefont {P.}~\bibnamefont {Bordia}},
  \bibinfo {author} {\bibfnamefont {H.~P.}\ \bibnamefont {L{\"u}schen}},
  \bibinfo {author} {\bibfnamefont {M.~H.}\ \bibnamefont {Fischer}}, \bibinfo
  {author} {\bibfnamefont {R.}~\bibnamefont {Vosk}}, \bibinfo {author}
  {\bibfnamefont {E.}~\bibnamefont {Altman}}, \bibinfo {author} {\bibfnamefont
  {U.}~\bibnamefont {Schneider}}, \ and\ \bibinfo {author} {\bibfnamefont
  {I.}~\bibnamefont {Bloch}},\ }\href {\doibase 10.1126/science.aaa7432}
  {\bibfield  {journal} {\bibinfo  {journal} {Science}\ }\textbf {\bibinfo
  {volume} {349}},\ \bibinfo {pages} {842} (\bibinfo {year}
  {2015})}\BibitemShut {NoStop}%
\bibitem [{\citenamefont {Nandkishore}\ and\ \citenamefont
  {Huse}(2015)}]{Nandkishore2015}%
  \BibitemOpen
  \bibfield  {author} {\bibinfo {author} {\bibfnamefont {R.}~\bibnamefont
  {Nandkishore}}\ and\ \bibinfo {author} {\bibfnamefont {D.~A.}\ \bibnamefont
  {Huse}},\ }\href {\doibase 10.1146/annurev-conmatphys-031214-014726}
  {\bibfield  {journal} {\bibinfo  {journal} {Ann. Rev. Cond. Matter Phys.}\
  }\textbf {\bibinfo {volume} {6}},\ \bibinfo {pages} {15} (\bibinfo {year}
  {2015})}\BibitemShut {NoStop}%
\bibitem [{\citenamefont {Serbyn}\ \emph {et~al.}(2014)\citenamefont {Serbyn},
  \citenamefont {Papi\ifmmode~\acute{c}\else \'{c}\fi{}},\ and\ \citenamefont
  {Abanin}}]{Serbyn2014}%
  \BibitemOpen
  \bibfield  {author} {\bibinfo {author} {\bibfnamefont {M.}~\bibnamefont
  {Serbyn}}, \bibinfo {author} {\bibfnamefont {Z.}~\bibnamefont
  {Papi\ifmmode~\acute{c}\else \'{c}\fi{}}}, \ and\ \bibinfo {author}
  {\bibfnamefont {D.~A.}\ \bibnamefont {Abanin}},\ }\href {\doibase
  10.1103/PhysRevB.90.174302} {\bibfield  {journal} {\bibinfo  {journal} {Phys.
  Rev. B}\ }\textbf {\bibinfo {volume} {90}},\ \bibinfo {pages} {174302}
  (\bibinfo {year} {2014})}\BibitemShut {NoStop}%
\bibitem [{\citenamefont {Deutsch}(1991)}]{Deutsch1991}%
  \BibitemOpen
  \bibfield  {author} {\bibinfo {author} {\bibfnamefont {J.~M.}\ \bibnamefont
  {Deutsch}},\ }\href {\doibase 10.1103/PhysRevA.43.2046} {\bibfield  {journal}
  {\bibinfo  {journal} {Phys. Rev. A}\ }\textbf {\bibinfo {volume} {43}},\
  \bibinfo {pages} {2046} (\bibinfo {year} {1991})}\BibitemShut {NoStop}%
\bibitem [{\citenamefont {Srednicki}(1994)}]{Srednicki1994}%
  \BibitemOpen
  \bibfield  {author} {\bibinfo {author} {\bibfnamefont {M.}~\bibnamefont
  {Srednicki}},\ }\href {\doibase 10.1103/PhysRevE.50.888} {\bibfield
  {journal} {\bibinfo  {journal} {Phys. Rev. E}\ }\textbf {\bibinfo {volume}
  {50}},\ \bibinfo {pages} {888} (\bibinfo {year} {1994})}\BibitemShut
  {NoStop}%
\bibitem [{\citenamefont {Rigol}\ \emph {et~al.}(2008)\citenamefont {Rigol},
  \citenamefont {Dunjko},\ and\ \citenamefont {Olshanii}}]{Rigol2008}%
  \BibitemOpen
  \bibfield  {author} {\bibinfo {author} {\bibfnamefont {M.}~\bibnamefont
  {Rigol}}, \bibinfo {author} {\bibfnamefont {V.}~\bibnamefont {Dunjko}}, \
  and\ \bibinfo {author} {\bibfnamefont {M.}~\bibnamefont {Olshanii}},\ }\href
  {\doibase 10.1038/nature06838} {\bibfield  {journal} {\bibinfo  {journal}
  {Nature}\ }\textbf {\bibinfo {volume} {452}},\ \bibinfo {pages} {854}
  (\bibinfo {year} {2008})}\BibitemShut {NoStop}%
\bibitem [{\citenamefont {Polkovnikov}\ \emph {et~al.}(2011)\citenamefont
  {Polkovnikov}, \citenamefont {Sengupta}, \citenamefont {Silva},\ and\
  \citenamefont {Vengalattore}}]{Polkovnikov2011}%
  \BibitemOpen
  \bibfield  {author} {\bibinfo {author} {\bibfnamefont {A.}~\bibnamefont
  {Polkovnikov}}, \bibinfo {author} {\bibfnamefont {K.}~\bibnamefont
  {Sengupta}}, \bibinfo {author} {\bibfnamefont {A.}~\bibnamefont {Silva}}, \
  and\ \bibinfo {author} {\bibfnamefont {M.}~\bibnamefont {Vengalattore}},\
  }\href {\doibase 10.1103/RevModPhys.83.863} {\bibfield  {journal} {\bibinfo
  {journal} {Rev. Mod. Phys.}\ }\textbf {\bibinfo {volume} {83}},\ \bibinfo
  {pages} {863} (\bibinfo {year} {2011})}\BibitemShut {NoStop}%
\bibitem [{\citenamefont {D'Alessio}\ \emph {et~al.}(2016)\citenamefont
  {D'Alessio}, \citenamefont {Kafri}, \citenamefont {Polkovnikov},\ and\
  \citenamefont {Rigol}}]{Rigol2016}%
  \BibitemOpen
  \bibfield  {author} {\bibinfo {author} {\bibfnamefont {L.}~\bibnamefont
  {D'Alessio}}, \bibinfo {author} {\bibfnamefont {Y.}~\bibnamefont {Kafri}},
  \bibinfo {author} {\bibfnamefont {A.}~\bibnamefont {Polkovnikov}}, \ and\
  \bibinfo {author} {\bibfnamefont {M.}~\bibnamefont {Rigol}},\ }\href
  {\doibase 10.1080/00018732.2016.1198134} {\bibfield  {journal} {\bibinfo
  {journal} {Adv. Phys.}\ }\textbf {\bibinfo {volume} {65}},\ \bibinfo {pages}
  {239} (\bibinfo {year} {2016})}\BibitemShut {NoStop}%
\bibitem [{sym()}]{symmetryNote}%
  \BibitemOpen
  \href@noop {} {}\bibinfo {note} {In case of global symmetries or
  super-selection sectors, each sector should be considered
  separately.}\BibitemShut {Stop}%
\bibitem [{\citenamefont {Farrelly}\ \emph {et~al.}(2017)\citenamefont
  {Farrelly}, \citenamefont {Brand{\~a}o},\ and\ \citenamefont
  {Cramer}}]{Farrelly2016}%
  \BibitemOpen
  \bibfield  {author} {\bibinfo {author} {\bibfnamefont {T.}~\bibnamefont
  {Farrelly}}, \bibinfo {author} {\bibfnamefont {F.~G.}\ \bibnamefont
  {Brand{\~a}o}}, \ and\ \bibinfo {author} {\bibfnamefont {M.}~\bibnamefont
  {Cramer}},\ }\href {\doibase 10.1103/PhysRevLett.118.140601} {\bibfield
  {journal} {\bibinfo  {journal} {Phys. Rev. Lett.}\ }\textbf {\bibinfo
  {volume} {118}},\ \bibinfo {pages} {140601} (\bibinfo {year}
  {2017})}\BibitemShut {NoStop}%
\bibitem [{\citenamefont {Gallego}\ \emph {et~al.}(2018)\citenamefont
  {Gallego}, \citenamefont {Wilming}, \citenamefont {Eisert},\ and\
  \citenamefont {Gogolin}}]{Gallego2017}%
  \BibitemOpen
  \bibfield  {author} {\bibinfo {author} {\bibfnamefont {R.}~\bibnamefont
  {Gallego}}, \bibinfo {author} {\bibfnamefont {H.}~\bibnamefont {Wilming}},
  \bibinfo {author} {\bibfnamefont {J.}~\bibnamefont {Eisert}}, \ and\ \bibinfo
  {author} {\bibfnamefont {C.}~\bibnamefont {Gogolin}},\ }\href {\doibase
  10.1103/PhysRevA.98.022135} {\bibfield  {journal} {\bibinfo  {journal} {Phys.
  Rev. A}\ }\textbf {\bibinfo {volume} {98}},\ \bibinfo {pages} {022135}
  (\bibinfo {year} {2018})}\BibitemShut {NoStop}%
\bibitem [{\citenamefont {Caux}\ and\ \citenamefont {Mossel}(2011)}]{Caux2010}%
  \BibitemOpen
  \bibfield  {author} {\bibinfo {author} {\bibfnamefont {J.-S.}\ \bibnamefont
  {Caux}}\ and\ \bibinfo {author} {\bibfnamefont {J.}~\bibnamefont {Mossel}},\
  }\href {\doibase 10.1088/1742-5468/2011/02/P02023} {\bibfield  {journal}
  {\bibinfo  {journal} {J. Stat. Mech.}\ ,\ \bibinfo {pages} {P02023}}
  (\bibinfo {year} {2011})}\BibitemShut {NoStop}%
\bibitem [{\citenamefont {Bennett}\ \emph {et~al.}(1996)\citenamefont
  {Bennett}, \citenamefont {Bernstein}, \citenamefont {Popescu},\ and\
  \citenamefont {Schumacher}}]{PhysRevA.53.2046}%
  \BibitemOpen
  \bibfield  {author} {\bibinfo {author} {\bibfnamefont {C.~H.}\ \bibnamefont
  {Bennett}}, \bibinfo {author} {\bibfnamefont {H.~J.}\ \bibnamefont
  {Bernstein}}, \bibinfo {author} {\bibfnamefont {S.}~\bibnamefont {Popescu}},
  \ and\ \bibinfo {author} {\bibfnamefont {B.}~\bibnamefont {Schumacher}},\
  }\href {\doibase 10.1103/PhysRevA.53.2046} {\bibfield  {journal} {\bibinfo
  {journal} {Phys. Rev. A}\ }\textbf {\bibinfo {volume} {53}},\ \bibinfo
  {pages} {2046} (\bibinfo {year} {1996})}\BibitemShut {NoStop}%
\bibitem [{\citenamefont {Garrison}\ and\ \citenamefont
  {Grover}(2018)}]{Garrison2015}%
  \BibitemOpen
  \bibfield  {author} {\bibinfo {author} {\bibfnamefont {J.~R.}\ \bibnamefont
  {Garrison}}\ and\ \bibinfo {author} {\bibfnamefont {T.}~\bibnamefont
  {Grover}},\ }\href {\doibase 10.1103/PhysRevX.8.021026} {\bibfield  {journal}
  {\bibinfo  {journal} {Phys. Rev. X}\ }\textbf {\bibinfo {volume} {8}},\
  \bibinfo {pages} {021026} (\bibinfo {year} {2018})}\BibitemShut {NoStop}%
\bibitem [{\citenamefont {Nakagawa}\ \emph {et~al.}(2018)\citenamefont
  {Nakagawa}, \citenamefont {Watanabe}, \citenamefont {Fujita},\ and\
  \citenamefont {Sugiura}}]{Nakagawa2018}%
  \BibitemOpen
  \bibfield  {author} {\bibinfo {author} {\bibfnamefont {Y.~O.}\ \bibnamefont
  {Nakagawa}}, \bibinfo {author} {\bibfnamefont {M.}~\bibnamefont {Watanabe}},
  \bibinfo {author} {\bibfnamefont {H.}~\bibnamefont {Fujita}}, \ and\ \bibinfo
  {author} {\bibfnamefont {S.}~\bibnamefont {Sugiura}},\ }\href {\doibase
  10.1038/s41467-018-03883-9} {\bibfield  {journal} {\bibinfo  {journal} {Nat.
  Commun.}\ }\textbf {\bibinfo {volume} {9}} (\bibinfo {year} {2018}),\
  10.1038/s41467-018-03883-9}\BibitemShut {NoStop}%
\bibitem [{\citenamefont {Lu}\ and\ \citenamefont {Grover}(2019)}]{Lu2017}%
  \BibitemOpen
  \bibfield  {author} {\bibinfo {author} {\bibfnamefont {T.-C.}\ \bibnamefont
  {Lu}}\ and\ \bibinfo {author} {\bibfnamefont {T.}~\bibnamefont {Grover}},\
  }\href {\doibase 10.1103/physreve.99.032111} {\bibfield  {journal} {\bibinfo
  {journal} {Phys. Rev. E}\ }\textbf {\bibinfo {volume} {99}} (\bibinfo {year}
  {2019}),\ 10.1103/physreve.99.032111}\BibitemShut {NoStop}%
\bibitem [{\citenamefont {Huang}(2019)}]{Huang2017}%
  \BibitemOpen
  \bibfield  {author} {\bibinfo {author} {\bibfnamefont {Y.}~\bibnamefont
  {Huang}},\ }\href {\doibase 10.1016/j.nuclphysb.2018.09.013} {\bibfield
  {journal} {\bibinfo  {journal} {Nucl. Phys. B}\ }\textbf {\bibinfo {volume}
  {938}},\ \bibinfo {pages} {594} (\bibinfo {year} {2019})}\BibitemShut
  {NoStop}%
\bibitem [{\citenamefont {Vidmar}\ \emph {et~al.}(2017)\citenamefont {Vidmar},
  \citenamefont {Hackl}, \citenamefont {Bianchi},\ and\ \citenamefont
  {Rigol}}]{Vidmar2017}%
  \BibitemOpen
  \bibfield  {author} {\bibinfo {author} {\bibfnamefont {L.}~\bibnamefont
  {Vidmar}}, \bibinfo {author} {\bibfnamefont {L.}~\bibnamefont {Hackl}},
  \bibinfo {author} {\bibfnamefont {E.}~\bibnamefont {Bianchi}}, \ and\
  \bibinfo {author} {\bibfnamefont {M.}~\bibnamefont {Rigol}},\ }\href
  {\doibase 10.1103/PhysRevLett.119.020601} {\bibfield  {journal} {\bibinfo
  {journal} {Phys. Rev. Lett.}\ }\textbf {\bibinfo {volume} {119}},\ \bibinfo
  {pages} {020601} (\bibinfo {year} {2017})}\BibitemShut {NoStop}%
\bibitem [{\citenamefont {Vidmar}\ and\ \citenamefont
  {Rigol}(2017)}]{Vidmar2017a}%
  \BibitemOpen
  \bibfield  {author} {\bibinfo {author} {\bibfnamefont {L.}~\bibnamefont
  {Vidmar}}\ and\ \bibinfo {author} {\bibfnamefont {M.}~\bibnamefont {Rigol}},\
  }\href {\doibase 10.1103/PhysRevLett.119.220603} {\bibfield  {journal}
  {\bibinfo  {journal} {Phys. Rev. Lett.}\ }\textbf {\bibinfo {volume} {119}},\
  \bibinfo {pages} {220603} (\bibinfo {year} {2017})}\BibitemShut {NoStop}%
\bibitem [{\citenamefont {Goold}\ \emph {et~al.}(2015)\citenamefont {Goold},
  \citenamefont {Gogolin}, \citenamefont {Clark}, \citenamefont {Eisert},
  \citenamefont {Scardicchio},\ and\ \citenamefont
  {Silva}}]{PhysRevB.92.180202}%
  \BibitemOpen
  \bibfield  {author} {\bibinfo {author} {\bibfnamefont {J.}~\bibnamefont
  {Goold}}, \bibinfo {author} {\bibfnamefont {C.}~\bibnamefont {Gogolin}},
  \bibinfo {author} {\bibfnamefont {S.~R.}\ \bibnamefont {Clark}}, \bibinfo
  {author} {\bibfnamefont {J.}~\bibnamefont {Eisert}}, \bibinfo {author}
  {\bibfnamefont {A.}~\bibnamefont {Scardicchio}}, \ and\ \bibinfo {author}
  {\bibfnamefont {A.}~\bibnamefont {Silva}},\ }\href {\doibase
  10.1103/PhysRevB.92.180202} {\bibfield  {journal} {\bibinfo  {journal} {Phys.
  Rev. B}\ }\textbf {\bibinfo {volume} {92}},\ \bibinfo {pages} {180202}
  (\bibinfo {year} {2015})}\BibitemShut {NoStop}%
\bibitem [{\citenamefont {Berry}(1985)}]{Berry1985}%
  \BibitemOpen
  \bibfield  {author} {\bibinfo {author} {\bibfnamefont {M.~V.}\ \bibnamefont
  {Berry}},\ }\href {\doibase 10.1098/rspa.1985.0078} {\bibfield  {journal}
  {\bibinfo  {journal} {Proc. Roy. Soc. A}\ }\textbf {\bibinfo {volume}
  {400}},\ \bibinfo {pages} {229} (\bibinfo {year} {1985})}\BibitemShut
  {NoStop}%
\bibitem [{\citenamefont {Prosen}(1999)}]{Prosen1999}%
  \BibitemOpen
  \bibfield  {author} {\bibinfo {author} {\bibfnamefont {T.}~\bibnamefont
  {Prosen}},\ }\href {\doibase 10.1103/PhysRevE.60.3949} {\bibfield  {journal}
  {\bibinfo  {journal} {Phys. Rev. E}\ }\textbf {\bibinfo {volume} {60}},\
  \bibinfo {pages} {3949} (\bibinfo {year} {1999})}\BibitemShut {NoStop}%
\bibitem [{\citenamefont {M{\"u}ller}\ \emph {et~al.}(2004)\citenamefont
  {M{\"u}ller}, \citenamefont {Heusler}, \citenamefont {Braun}, \citenamefont
  {Haake},\ and\ \citenamefont {Altland}}]{Mueller2004}%
  \BibitemOpen
  \bibfield  {author} {\bibinfo {author} {\bibfnamefont {S.}~\bibnamefont
  {M{\"u}ller}}, \bibinfo {author} {\bibfnamefont {S.}~\bibnamefont {Heusler}},
  \bibinfo {author} {\bibfnamefont {P.}~\bibnamefont {Braun}}, \bibinfo
  {author} {\bibfnamefont {F.}~\bibnamefont {Haake}}, \ and\ \bibinfo {author}
  {\bibfnamefont {A.}~\bibnamefont {Altland}},\ }\href {\doibase
  10.1103/PhysRevLett.93.014103} {\bibfield  {journal} {\bibinfo  {journal}
  {Phys. Rev. Lett.}\ }\textbf {\bibinfo {volume} {93}},\ \bibinfo {pages}
  {014103} (\bibinfo {year} {2004})}\BibitemShut {NoStop}%
\bibitem [{\citenamefont {Kollath}\ \emph {et~al.}(2010)\citenamefont
  {Kollath}, \citenamefont {Roux}, \citenamefont {Biroli},\ and\ \citenamefont
  {L{\"a}uchli}}]{Kollath2010}%
  \BibitemOpen
  \bibfield  {author} {\bibinfo {author} {\bibfnamefont {C.}~\bibnamefont
  {Kollath}}, \bibinfo {author} {\bibfnamefont {G.}~\bibnamefont {Roux}},
  \bibinfo {author} {\bibfnamefont {G.}~\bibnamefont {Biroli}}, \ and\ \bibinfo
  {author} {\bibfnamefont {A.~M.}\ \bibnamefont {L{\"a}uchli}},\ }\href
  {\doibase 10.1088/1742-5468/2010/08/P08011} {\bibfield  {journal} {\bibinfo
  {journal} {J. Stat. Mech.}\ }\textbf {\bibinfo {volume} {2010}},\ \bibinfo
  {pages} {P08011} (\bibinfo {year} {2010})}\BibitemShut {NoStop}%
\bibitem [{\citenamefont {Keating}\ \emph {et~al.}(2015)\citenamefont
  {Keating}, \citenamefont {Linden},\ and\ \citenamefont
  {Wells}}]{Keating2015}%
  \BibitemOpen
  \bibfield  {author} {\bibinfo {author} {\bibfnamefont {J.~P.}\ \bibnamefont
  {Keating}}, \bibinfo {author} {\bibfnamefont {N.}~\bibnamefont {Linden}}, \
  and\ \bibinfo {author} {\bibfnamefont {H.~J.}\ \bibnamefont {Wells}},\ }\href
  {\doibase 10.1007/s00220-015-2366-0} {\bibfield  {journal} {\bibinfo
  {journal} {Commun. Math. Phys.}\ }\textbf {\bibinfo {volume} {338}},\
  \bibinfo {pages} {81} (\bibinfo {year} {2015})}\BibitemShut {NoStop}%
\bibitem [{\citenamefont {Kos}\ \emph {et~al.}(2018)\citenamefont {Kos},
  \citenamefont {Ljubotina},\ and\ \citenamefont {Prosen}}]{Kos2017}%
  \BibitemOpen
  \bibfield  {author} {\bibinfo {author} {\bibfnamefont {P.}~\bibnamefont
  {Kos}}, \bibinfo {author} {\bibfnamefont {M.}~\bibnamefont {Ljubotina}}, \
  and\ \bibinfo {author} {\bibfnamefont {T.}~\bibnamefont {Prosen}},\ }\href
  {\doibase 10.1103/physrevx.8.021062} {\bibfield  {journal} {\bibinfo
  {journal} {Phys. Rev. X}\ }\textbf {\bibinfo {volume} {8}} (\bibinfo {year}
  {2018}),\ 10.1103/physrevx.8.021062}\BibitemShut {NoStop}%
\bibitem [{\citenamefont {Lubkin}(1978)}]{Lubkin1978}%
  \BibitemOpen
  \bibfield  {author} {\bibinfo {author} {\bibfnamefont {E.}~\bibnamefont
  {Lubkin}},\ }\href {\doibase 10.1063/1.523763} {\bibfield  {journal}
  {\bibinfo  {journal} {J. Math. Phys.}\ }\textbf {\bibinfo {volume} {19}},\
  \bibinfo {pages} {1028} (\bibinfo {year} {1978})}\BibitemShut {NoStop}%
\bibitem [{\citenamefont {Lloyd}\ and\ \citenamefont
  {Pagels}(1988)}]{Lloyd1988}%
  \BibitemOpen
  \bibfield  {author} {\bibinfo {author} {\bibfnamefont {S.}~\bibnamefont
  {Lloyd}}\ and\ \bibinfo {author} {\bibfnamefont {H.}~\bibnamefont {Pagels}},\
  }\href {\doibase doi.org/10.1016/0003-4916(88)90094-2} {\bibfield  {journal}
  {\bibinfo  {journal} {Ann. Phys.}\ }\textbf {\bibinfo {volume} {188}},\
  \bibinfo {pages} {186} (\bibinfo {year} {1988})}\BibitemShut {NoStop}%
\bibitem [{\citenamefont {Page}(1993)}]{Page1993}%
  \BibitemOpen
  \bibfield  {author} {\bibinfo {author} {\bibfnamefont {D.~N.}\ \bibnamefont
  {Page}},\ }\href {\doibase 10.1103/PhysRevLett.71.1291} {\bibfield  {journal}
  {\bibinfo  {journal} {Phys. Rev. Lett.}\ }\textbf {\bibinfo {volume} {71}},\
  \bibinfo {pages} {1291} (\bibinfo {year} {1993})}\BibitemShut {NoStop}%
\bibitem [{\citenamefont {Nadal}\ \emph {et~al.}(2010)\citenamefont {Nadal},
  \citenamefont {Majumdar},\ and\ \citenamefont {Vergassola}}]{Nadal2010}%
  \BibitemOpen
  \bibfield  {author} {\bibinfo {author} {\bibfnamefont {C.}~\bibnamefont
  {Nadal}}, \bibinfo {author} {\bibfnamefont {S.~N.}\ \bibnamefont {Majumdar}},
  \ and\ \bibinfo {author} {\bibfnamefont {M.}~\bibnamefont {Vergassola}},\
  }\href {\doibase 10.1103/PhysRevLett.104.110501} {\bibfield  {journal}
  {\bibinfo  {journal} {Phys. Rev. Lett.}\ }\textbf {\bibinfo {volume} {104}},\
  \bibinfo {pages} {110501} (\bibinfo {year} {2010})}\BibitemShut {NoStop}%
\bibitem [{\citenamefont {Nadal}\ \emph {et~al.}(2011)\citenamefont {Nadal},
  \citenamefont {Majumdar},\ and\ \citenamefont {Vergassola}}]{Nadal2011}%
  \BibitemOpen
  \bibfield  {author} {\bibinfo {author} {\bibfnamefont {C.}~\bibnamefont
  {Nadal}}, \bibinfo {author} {\bibfnamefont {S.~N.}\ \bibnamefont {Majumdar}},
  \ and\ \bibinfo {author} {\bibfnamefont {M.}~\bibnamefont {Vergassola}},\
  }\href {\doibase 10.1007/s10955-010-0108-4} {\bibfield  {journal} {\bibinfo
  {journal} {Journal of Statistical Physics}\ }\textbf {\bibinfo {volume}
  {142}},\ \bibinfo {pages} {403} (\bibinfo {year} {2011})}\BibitemShut
  {NoStop}%
\bibitem [{sym()}]{HaarNote}%
  \BibitemOpen
  \href@noop {} {}\bibinfo {note} {While it is immediate that each eigenstate of a random Hamiltonian fulfills a volume law with high probability, the distribution of all the basis states is not independent. It is thus not immediate that the probability that all $d^N$ eigenstates fulfill a volume law is also arbitrarily close to unity for large systems. However, the probability for each of them not to follow a volume law is exponentially small in $d^N$ according to Ref.~\cite{Nadal2010,Nadal2011}, while there are only $d^N$ state, which indicates that even this latter probability is very large.}\BibitemShut {Stop}%
\bibitem [{\citenamefont {Eisert}\ \emph {et~al.}(2010)\citenamefont {Eisert},
  \citenamefont {Cramer},\ and\ \citenamefont {Plenio}}]{AreaReview}%
  \BibitemOpen
  \bibfield  {author} {\bibinfo {author} {\bibfnamefont {J.}~\bibnamefont
  {Eisert}}, \bibinfo {author} {\bibfnamefont {M.}~\bibnamefont {Cramer}}, \
  and\ \bibinfo {author} {\bibfnamefont {M.~B.}\ \bibnamefont {Plenio}},\
  }\href {\doibase 10.1103/RevModPhys.82.277} {\bibfield  {journal} {\bibinfo
  {journal} {Rev. Mod. Phys.}\ }\textbf {\bibinfo {volume} {82}},\ \bibinfo
  {pages} {277} (\bibinfo {year} {2010})}\BibitemShut {NoStop}%
\bibitem [{\citenamefont {Bauer}\ and\ \citenamefont {Nayak}(2013)}]{Bauer}%
  \BibitemOpen
  \bibfield  {author} {\bibinfo {author} {\bibfnamefont {B.}~\bibnamefont
  {Bauer}}\ and\ \bibinfo {author} {\bibfnamefont {C.}~\bibnamefont {Nayak}},\
  }\href {\doibase 10.1088/1742-5468/2013/09/P09005} {\bibfield  {journal}
  {\bibinfo  {journal} {J. Stat. Mech.}\ }\textbf {\bibinfo {volume} {2013}},\
  \bibinfo {pages} {P09005} (\bibinfo {year} {2013})}\BibitemShut {NoStop}%
\bibitem [{\citenamefont {Friesdorf}\ \emph {et~al.}(2015)\citenamefont
  {Friesdorf}, \citenamefont {Werner}, \citenamefont {Brown}, \citenamefont
  {Scholz},\ and\ \citenamefont {Eisert}}]{1409.1252}%
  \BibitemOpen
  \bibfield  {author} {\bibinfo {author} {\bibfnamefont {M.}~\bibnamefont
  {Friesdorf}}, \bibinfo {author} {\bibfnamefont {A.~H.}\ \bibnamefont
  {Werner}}, \bibinfo {author} {\bibfnamefont {W.}~\bibnamefont {Brown}},
  \bibinfo {author} {\bibfnamefont {V.~B.}\ \bibnamefont {Scholz}}, \ and\
  \bibinfo {author} {\bibfnamefont {J.}~\bibnamefont {Eisert}},\ }\href
  {\doibase 10.1103/PhysRevLett.114.170505} {\bibfield  {journal} {\bibinfo
  {journal} {Phys. Rev. Lett.}\ }\textbf {\bibinfo {volume} {114}},\ \bibinfo
  {pages} {170505} (\bibinfo {year} {2015})}\BibitemShut {NoStop}%
\bibitem [{\citenamefont {Anshu}(2016)}]{Anshu2016}%
  \BibitemOpen
  \bibfield  {author} {\bibinfo {author} {\bibfnamefont {A.}~\bibnamefont
  {Anshu}},\ }\href {\doibase 10.1088/1367-2630/18/8/083011} {\bibfield
  {journal} {\bibinfo  {journal} {New J. Phys.}\ }\textbf {\bibinfo {volume}
  {18}},\ \bibinfo {pages} {083011} (\bibinfo {year} {2016})}\BibitemShut
  {NoStop}%
\bibitem [{\citenamefont {Bernien}\ \emph {et~al.}(2017)\citenamefont
  {Bernien}, \citenamefont {Schwartz}, \citenamefont {Keesling}, \citenamefont
  {Levine}, \citenamefont {Omran}, \citenamefont {Pichler}, \citenamefont
  {Choi}, \citenamefont {Zibrov}, \citenamefont {Endres}, \citenamefont
  {Greiner}, \citenamefont {Vuleti{\'{c}}},\ and\ \citenamefont
  {Lukin}}]{Bernien2017}%
  \BibitemOpen
  \bibfield  {author} {\bibinfo {author} {\bibfnamefont {H.}~\bibnamefont
  {Bernien}}, \bibinfo {author} {\bibfnamefont {S.}~\bibnamefont {Schwartz}},
  \bibinfo {author} {\bibfnamefont {A.}~\bibnamefont {Keesling}}, \bibinfo
  {author} {\bibfnamefont {H.}~\bibnamefont {Levine}}, \bibinfo {author}
  {\bibfnamefont {A.}~\bibnamefont {Omran}}, \bibinfo {author} {\bibfnamefont
  {H.}~\bibnamefont {Pichler}}, \bibinfo {author} {\bibfnamefont
  {S.}~\bibnamefont {Choi}}, \bibinfo {author} {\bibfnamefont {A.~S.}\
  \bibnamefont {Zibrov}}, \bibinfo {author} {\bibfnamefont {M.}~\bibnamefont
  {Endres}}, \bibinfo {author} {\bibfnamefont {M.}~\bibnamefont {Greiner}},
  \bibinfo {author} {\bibfnamefont {V.}~\bibnamefont {Vuleti{\'{c}}}}, \ and\
  \bibinfo {author} {\bibfnamefont {M.~D.}\ \bibnamefont {Lukin}},\ }\href
  {\doibase 10.1038/nature24622} {\bibfield  {journal} {\bibinfo  {journal}
  {Nature}\ }\textbf {\bibinfo {volume} {551}},\ \bibinfo {pages} {579}
  (\bibinfo {year} {2017})}\BibitemShut {NoStop}%
\bibitem [{\citenamefont {Turner}\ \emph
  {et~al.}(2018{\natexlab{a}})\citenamefont {Turner}, \citenamefont
  {Michailidis}, \citenamefont {Abanin}, \citenamefont {Serbyn},\ and\
  \citenamefont {Papi{\'{c}}}}]{Turner2018}%
  \BibitemOpen
  \bibfield  {author} {\bibinfo {author} {\bibfnamefont {C.~J.}\ \bibnamefont
  {Turner}}, \bibinfo {author} {\bibfnamefont {A.~A.}\ \bibnamefont
  {Michailidis}}, \bibinfo {author} {\bibfnamefont {D.~A.}\ \bibnamefont
  {Abanin}}, \bibinfo {author} {\bibfnamefont {M.}~\bibnamefont {Serbyn}}, \
  and\ \bibinfo {author} {\bibfnamefont {Z.}~\bibnamefont {Papi{\'{c}}}},\
  }\href {\doibase 10.1038/s41567-018-0137-5} {\bibfield  {journal} {\bibinfo
  {journal} {Nat. Phys.}\ }\textbf {\bibinfo {volume} {14}},\ \bibinfo {pages}
  {745} (\bibinfo {year} {2018}{\natexlab{a}})}\BibitemShut {NoStop}%
\bibitem [{\citenamefont {Moudgalya}\ \emph
  {et~al.}(2018{\natexlab{a}})\citenamefont {Moudgalya}, \citenamefont
  {Rachel}, \citenamefont {Bernevig},\ and\ \citenamefont
  {Regnault}}]{Moudgalya2018}%
  \BibitemOpen
  \bibfield  {author} {\bibinfo {author} {\bibfnamefont {S.}~\bibnamefont
  {Moudgalya}}, \bibinfo {author} {\bibfnamefont {S.}~\bibnamefont {Rachel}},
  \bibinfo {author} {\bibfnamefont {B.~A.}\ \bibnamefont {Bernevig}}, \ and\
  \bibinfo {author} {\bibfnamefont {N.}~\bibnamefont {Regnault}},\ }\href
  {\doibase 10.1103/physrevb.98.235155} {\bibfield  {journal} {\bibinfo
  {journal} {Phys. Rev. B}\ }\textbf {\bibinfo {volume} {98}} (\bibinfo {year}
  {2018}{\natexlab{a}}),\ 10.1103/physrevb.98.235155}\BibitemShut {NoStop}%
\bibitem [{\citenamefont {Turner}\ \emph
  {et~al.}(2018{\natexlab{b}})\citenamefont {Turner}, \citenamefont
  {Michailidis}, \citenamefont {Abanin}, \citenamefont {Serbyn},\ and\
  \citenamefont {Papi{\'{c}}}}]{Turner2018a}%
  \BibitemOpen
  \bibfield  {author} {\bibinfo {author} {\bibfnamefont {C.~J.}\ \bibnamefont
  {Turner}}, \bibinfo {author} {\bibfnamefont {A.~A.}\ \bibnamefont
  {Michailidis}}, \bibinfo {author} {\bibfnamefont {D.~A.}\ \bibnamefont
  {Abanin}}, \bibinfo {author} {\bibfnamefont {M.}~\bibnamefont {Serbyn}}, \
  and\ \bibinfo {author} {\bibfnamefont {Z.}~\bibnamefont {Papi{\'{c}}}},\
  }\href {\doibase 10.1103/physrevb.98.155134} {\bibfield  {journal} {\bibinfo
  {journal} {Phys. Rev. B}\ }\textbf {\bibinfo {volume} {98}} (\bibinfo {year}
  {2018}{\natexlab{b}}),\ 10.1103/physrevb.98.155134}\BibitemShut {NoStop}%
\bibitem [{\citenamefont {Moudgalya}\ \emph
  {et~al.}(2018{\natexlab{b}})\citenamefont {Moudgalya}, \citenamefont
  {Regnault},\ and\ \citenamefont {Bernevig}}]{Moudgalya2018a}%
  \BibitemOpen
  \bibfield  {author} {\bibinfo {author} {\bibfnamefont {S.}~\bibnamefont
  {Moudgalya}}, \bibinfo {author} {\bibfnamefont {N.}~\bibnamefont {Regnault}},
  \ and\ \bibinfo {author} {\bibfnamefont {B.~A.}\ \bibnamefont {Bernevig}},\
  }\href {\doibase 10.1103/physrevb.98.235156} {\bibfield  {journal} {\bibinfo
  {journal} {Phys. Rev. B}\ }\textbf {\bibinfo {volume} {98}} (\bibinfo {year}
  {2018}{\natexlab{b}}),\ 10.1103/physrevb.98.235156}\BibitemShut {NoStop}%
\bibitem [{\citenamefont {Choi}\ \emph {et~al.}(2018)\citenamefont {Choi},
  \citenamefont {Turner}, \citenamefont {Pichler}, \citenamefont {Ho},
  \citenamefont {Michailidis}, \citenamefont {Papić}, \citenamefont {Serbyn},
  \citenamefont {Lukin},\ and\ \citenamefont {Abanin}}]{Choi2018}%
  \BibitemOpen
  \bibfield  {author} {\bibinfo {author} {\bibfnamefont {S.}~\bibnamefont
  {Choi}}, \bibinfo {author} {\bibfnamefont {C.~J.}\ \bibnamefont {Turner}},
  \bibinfo {author} {\bibfnamefont {H.}~\bibnamefont {Pichler}}, \bibinfo
  {author} {\bibfnamefont {W.~W.}\ \bibnamefont {Ho}}, \bibinfo {author}
  {\bibfnamefont {A.~A.}\ \bibnamefont {Michailidis}}, \bibinfo {author}
  {\bibfnamefont {Z.}~\bibnamefont {Papić}}, \bibinfo {author} {\bibfnamefont
  {M.}~\bibnamefont {Serbyn}}, \bibinfo {author} {\bibfnamefont {M.~D.}\
  \bibnamefont {Lukin}}, \ and\ \bibinfo {author} {\bibfnamefont {D.~A.}\
  \bibnamefont {Abanin}},\ }\href@noop {} {\enquote {\bibinfo {title} {Emergent
  su(2) dynamics and perfect quantum many-body scars},}\ } (\bibinfo {year}
  {2018}),\ \Eprint {http://arxiv.org/abs/1812.05561v1} {1812.05561v1}
  \BibitemShut {NoStop}%
\bibitem [{\citenamefont {Calabrese}\ and\ \citenamefont
  {Cardy}(2006)}]{Calabrese2006}%
  \BibitemOpen
  \bibfield  {author} {\bibinfo {author} {\bibfnamefont {P.}~\bibnamefont
  {Calabrese}}\ and\ \bibinfo {author} {\bibfnamefont {J.}~\bibnamefont
  {Cardy}},\ }\href {\doibase 10.1103/PhysRevLett.96.136801} {\bibfield
  {journal} {\bibinfo  {journal} {Phys. Rev. Lett.}\ }\textbf {\bibinfo
  {volume} {96}},\ \bibinfo {pages} {136801} (\bibinfo {year}
  {2006})}\BibitemShut {NoStop}%
\bibitem [{\citenamefont {Calabrese}\ and\ \citenamefont
  {Caux}(2007)}]{Calabrese2007}%
  \BibitemOpen
  \bibfield  {author} {\bibinfo {author} {\bibfnamefont {P.}~\bibnamefont
  {Calabrese}}\ and\ \bibinfo {author} {\bibfnamefont {J.-S.}\ \bibnamefont
  {Caux}},\ }\href {\doibase 10.1088/1742-5468/2007/08/P08032} {\bibfield
  {journal} {\bibinfo  {journal} {J. Stat. Mech.}\ }\textbf {\bibinfo {volume}
  {2007}},\ \bibinfo {pages} {P08032} (\bibinfo {year} {2007})}\BibitemShut
  {NoStop}%
\bibitem [{\citenamefont {Eisler}\ and\ \citenamefont
  {Peschel}(2007)}]{Eisler2007}%
  \BibitemOpen
  \bibfield  {author} {\bibinfo {author} {\bibfnamefont {V.}~\bibnamefont
  {Eisler}}\ and\ \bibinfo {author} {\bibfnamefont {I.}~\bibnamefont
  {Peschel}},\ }\href {\doibase 10.1088/1742-5468/2007/06/P06005} {\bibfield
  {journal} {\bibinfo  {journal} {J. Stat. Mech: Theory Exp.}\ }\textbf
  {\bibinfo {volume} {2007}},\ \bibinfo {pages} {P06005} (\bibinfo {year}
  {2007})}\BibitemShut {NoStop}%
\bibitem [{\citenamefont {Rigol}\ \emph {et~al.}(2007)\citenamefont {Rigol},
  \citenamefont {Dunjko}, \citenamefont {Yurovsky},\ and\ \citenamefont
  {Olshanii}}]{Rigol2007}%
  \BibitemOpen
  \bibfield  {author} {\bibinfo {author} {\bibfnamefont {M.}~\bibnamefont
  {Rigol}}, \bibinfo {author} {\bibfnamefont {V.}~\bibnamefont {Dunjko}},
  \bibinfo {author} {\bibfnamefont {V.}~\bibnamefont {Yurovsky}}, \ and\
  \bibinfo {author} {\bibfnamefont {M.}~\bibnamefont {Olshanii}},\ }\href
  {\doibase 10.1103/PhysRevLett.98.050405} {\bibfield  {journal} {\bibinfo
  {journal} {Phys. Rev. Lett.}\ }\textbf {\bibinfo {volume} {98}},\ \bibinfo
  {pages} {050405} (\bibinfo {year} {2007})}\BibitemShut {NoStop}%
\bibitem [{\citenamefont {Cramer}\ \emph
  {et~al.}(2008{\natexlab{a}})\citenamefont {Cramer}, \citenamefont {Dawson},
  \citenamefont {Eisert},\ and\ \citenamefont {Osborne}}]{Cramer2008}%
  \BibitemOpen
  \bibfield  {author} {\bibinfo {author} {\bibfnamefont {M.}~\bibnamefont
  {Cramer}}, \bibinfo {author} {\bibfnamefont {C.~M.}\ \bibnamefont {Dawson}},
  \bibinfo {author} {\bibfnamefont {J.}~\bibnamefont {Eisert}}, \ and\ \bibinfo
  {author} {\bibfnamefont {T.~J.}\ \bibnamefont {Osborne}},\ }\href {\doibase
  10.1103/PhysRevLett.100.030602} {\bibfield  {journal} {\bibinfo  {journal}
  {Phys. Rev. Lett.}\ }\textbf {\bibinfo {volume} {100}},\ \bibinfo {pages}
  {030602} (\bibinfo {year} {2008}{\natexlab{a}})}\BibitemShut {NoStop}%
\bibitem [{\citenamefont {Cramer}\ \emph
  {et~al.}(2008{\natexlab{b}})\citenamefont {Cramer}, \citenamefont {Flesch},
  \citenamefont {McCulloch}, \citenamefont {Schollw{\"o}ck},\ and\
  \citenamefont {Eisert}}]{Cramer2008a}%
  \BibitemOpen
  \bibfield  {author} {\bibinfo {author} {\bibfnamefont {M.}~\bibnamefont
  {Cramer}}, \bibinfo {author} {\bibfnamefont {A.}~\bibnamefont {Flesch}},
  \bibinfo {author} {\bibfnamefont {I.~P.}\ \bibnamefont {McCulloch}}, \bibinfo
  {author} {\bibfnamefont {U.}~\bibnamefont {Schollw{\"o}ck}}, \ and\ \bibinfo
  {author} {\bibfnamefont {J.}~\bibnamefont {Eisert}},\ }\href {\doibase
  10.1103/PhysRevLett.101.063001} {\bibfield  {journal} {\bibinfo  {journal}
  {Phys. Rev. Lett.}\ }\textbf {\bibinfo {volume} {101}},\ \bibinfo {pages}
  {063001} (\bibinfo {year} {2008}{\natexlab{b}})}\BibitemShut {NoStop}%
\bibitem [{\citenamefont {Barthel}\ and\ \citenamefont
  {Schollw{\"o}ck}(2008)}]{Barthel2008}%
  \BibitemOpen
  \bibfield  {author} {\bibinfo {author} {\bibfnamefont {T.}~\bibnamefont
  {Barthel}}\ and\ \bibinfo {author} {\bibfnamefont {U.}~\bibnamefont
  {Schollw{\"o}ck}},\ }\href {\doibase 10.1103/PhysRevLett.100.100601}
  {\bibfield  {journal} {\bibinfo  {journal} {Phys. Rev. Lett.}\ }\textbf
  {\bibinfo {volume} {100}},\ \bibinfo {pages} {100601} (\bibinfo {year}
  {2008})}\BibitemShut {NoStop}%
\bibitem [{\citenamefont {Flesch}\ \emph {et~al.}(2008)\citenamefont {Flesch},
  \citenamefont {Cramer}, \citenamefont {McCulloch}, \citenamefont
  {Schollwoeck},\ and\ \citenamefont {Eisert}}]{Flesch2008}%
  \BibitemOpen
  \bibfield  {author} {\bibinfo {author} {\bibfnamefont {A.}~\bibnamefont
  {Flesch}}, \bibinfo {author} {\bibfnamefont {M.}~\bibnamefont {Cramer}},
  \bibinfo {author} {\bibfnamefont {I.~P.}\ \bibnamefont {McCulloch}}, \bibinfo
  {author} {\bibfnamefont {U.}~\bibnamefont {Schollwoeck}}, \ and\ \bibinfo
  {author} {\bibfnamefont {J.}~\bibnamefont {Eisert}},\ }\href {\doibase
  10.1103/PhysRevA.78.033608} {\bibfield  {journal} {\bibinfo  {journal} {Phys.
  Rev. A}\ }\textbf {\bibinfo {volume} {78}},\ \bibinfo {pages} {033608}
  (\bibinfo {year} {2008})}\BibitemShut {NoStop}%
\bibitem [{\citenamefont {Calabrese}\ \emph {et~al.}(2011)\citenamefont
  {Calabrese}, \citenamefont {Essler},\ and\ \citenamefont
  {Fagotti}}]{Calabrese2011}%
  \BibitemOpen
  \bibfield  {author} {\bibinfo {author} {\bibfnamefont {P.}~\bibnamefont
  {Calabrese}}, \bibinfo {author} {\bibfnamefont {F.~H.~L.}\ \bibnamefont
  {Essler}}, \ and\ \bibinfo {author} {\bibfnamefont {M.}~\bibnamefont
  {Fagotti}},\ }\href {\doibase 10.1103/PhysRevLett.106.227203} {\bibfield
  {journal} {\bibinfo  {journal} {Phys. Rev. Lett.}\ }\textbf {\bibinfo
  {volume} {106}},\ \bibinfo {pages} {227203} (\bibinfo {year}
  {2011})}\BibitemShut {NoStop}%
\bibitem [{\citenamefont {Caux}\ and\ \citenamefont {Essler}(2013)}]{Caux2013}%
  \BibitemOpen
  \bibfield  {author} {\bibinfo {author} {\bibfnamefont {J.-S.}\ \bibnamefont
  {Caux}}\ and\ \bibinfo {author} {\bibfnamefont {F.~H.~L.}\ \bibnamefont
  {Essler}},\ }\href {\doibase 10.1103/PhysRevLett.110.257203} {\bibfield
  {journal} {\bibinfo  {journal} {Phys. Rev. Lett.}\ }\textbf {\bibinfo
  {volume} {110}},\ \bibinfo {pages} {257203} (\bibinfo {year}
  {2013})}\BibitemShut {NoStop}%
\bibitem [{\citenamefont {Goldstein}\ \emph {et~al.}(2013)\citenamefont
  {Goldstein}, \citenamefont {Hara},\ and\ \citenamefont
  {Tasaki}}]{Goldstein2013}%
  \BibitemOpen
  \bibfield  {author} {\bibinfo {author} {\bibfnamefont {S.}~\bibnamefont
  {Goldstein}}, \bibinfo {author} {\bibfnamefont {T.}~\bibnamefont {Hara}}, \
  and\ \bibinfo {author} {\bibfnamefont {H.}~\bibnamefont {Tasaki}},\ }\href
  {\doibase 10.1103/PhysRevLett.111.140401} {\bibfield  {journal} {\bibinfo
  {journal} {Phys. Rev. Lett.}\ }\textbf {\bibinfo {volume} {111}},\ \bibinfo
  {pages} {140401} (\bibinfo {year} {2013})}\BibitemShut {NoStop}%
\bibitem [{\citenamefont {Malabarba}\ \emph {et~al.}(2014)\citenamefont
  {Malabarba}, \citenamefont {Garc{\'{i}}a-Pintos}, \citenamefont {Linden},
  \citenamefont {Farrelly},\ and\ \citenamefont {Short}}]{Malabarba2014}%
  \BibitemOpen
  \bibfield  {author} {\bibinfo {author} {\bibfnamefont {A.~S.~L.}\
  \bibnamefont {Malabarba}}, \bibinfo {author} {\bibfnamefont {L.~P.}\
  \bibnamefont {Garc{\'{i}}a-Pintos}}, \bibinfo {author} {\bibfnamefont
  {N.}~\bibnamefont {Linden}}, \bibinfo {author} {\bibfnamefont {T.~C.}\
  \bibnamefont {Farrelly}}, \ and\ \bibinfo {author} {\bibfnamefont {A.~J.}\
  \bibnamefont {Short}},\ }\href {\doibase 10.1103/PhysRevE.90.012121}
  {\bibfield  {journal} {\bibinfo  {journal} {Phys. Rev. E}\ }\textbf {\bibinfo
  {volume} {90}},\ \bibinfo {pages} {012121} (\bibinfo {year} {2014})},\
  \Eprint {http://arxiv.org/abs/1402.1093v1} {arXiv:1402.1093v1} \BibitemShut
  {NoStop}%
\bibitem [{\citenamefont {Garc{\'\i}a-Pintos}\ \emph
  {et~al.}(2017)\citenamefont {Garc{\'\i}a-Pintos}, \citenamefont {Linden},
  \citenamefont {Malabarba}, \citenamefont {Short},\ and\ \citenamefont
  {Winter}}]{Garcia-Pintos2015a}%
  \BibitemOpen
  \bibfield  {author} {\bibinfo {author} {\bibfnamefont {L.~P.}\ \bibnamefont
  {Garc{\'\i}a-Pintos}}, \bibinfo {author} {\bibfnamefont {N.}~\bibnamefont
  {Linden}}, \bibinfo {author} {\bibfnamefont {A.~S.}\ \bibnamefont
  {Malabarba}}, \bibinfo {author} {\bibfnamefont {A.~J.}\ \bibnamefont
  {Short}}, \ and\ \bibinfo {author} {\bibfnamefont {A.}~\bibnamefont
  {Winter}},\ }\href {\doibase 10.1103/PhysRevX.7.031027} {\bibfield  {journal}
  {\bibinfo  {journal} {Phys. Rev. X}\ }\textbf {\bibinfo {volume} {7}},\
  \bibinfo {pages} {031027} (\bibinfo {year} {2017})}\BibitemShut {NoStop}%
\bibitem [{\citenamefont {Goldstein}\ \emph {et~al.}(2015)\citenamefont
  {Goldstein}, \citenamefont {Hara},\ and\ \citenamefont
  {Tasaki}}]{Goldstein2015}%
  \BibitemOpen
  \bibfield  {author} {\bibinfo {author} {\bibfnamefont {S.}~\bibnamefont
  {Goldstein}}, \bibinfo {author} {\bibfnamefont {T.}~\bibnamefont {Hara}}, \
  and\ \bibinfo {author} {\bibfnamefont {H.}~\bibnamefont {Tasaki}},\ }\href
  {\doibase 10.1088/1367-2630/17/4/045002} {\bibfield  {journal} {\bibinfo
  {journal} {New J. Phys.}\ }\textbf {\bibinfo {volume} {17}},\ \bibinfo
  {pages} {045002} (\bibinfo {year} {2015})}\BibitemShut {NoStop}%
\bibitem [{\citenamefont {Farrelly}(2016)}]{Farrelly2016a}%
  \BibitemOpen
  \bibfield  {author} {\bibinfo {author} {\bibfnamefont {T.}~\bibnamefont
  {Farrelly}},\ }\href {\doibase 10.1088/1367-2630/18/7/073014} {\bibfield
  {journal} {\bibinfo  {journal} {New J. Phys.}\ }\textbf {\bibinfo {volume}
  {18}},\ \bibinfo {pages} {073014} (\bibinfo {year} {2016})}\BibitemShut
  {NoStop}%
\bibitem [{\citenamefont {Reimann}(2016)}]{Reimann2016}%
  \BibitemOpen
  \bibfield  {author} {\bibinfo {author} {\bibfnamefont {P.}~\bibnamefont
  {Reimann}},\ }\href {\doibase 10.1038/ncomms10821} {\bibfield  {journal}
  {\bibinfo  {journal} {Nature Comm.}\ }\textbf {\bibinfo {volume} {7}},\
  \bibinfo {pages} {10821} (\bibinfo {year} {2016})}\BibitemShut {NoStop}%
\bibitem [{\citenamefont {de~Oliveira}\ \emph {et~al.}(2018)\citenamefont
  {de~Oliveira}, \citenamefont {Charalambous}, \citenamefont {Jonathan},
  \citenamefont {Lewenstein},\ and\ \citenamefont {Riera}}]{DeOliveira2017}%
  \BibitemOpen
  \bibfield  {author} {\bibinfo {author} {\bibfnamefont {T.~R.}\ \bibnamefont
  {de~Oliveira}}, \bibinfo {author} {\bibfnamefont {C.}~\bibnamefont
  {Charalambous}}, \bibinfo {author} {\bibfnamefont {D.}~\bibnamefont
  {Jonathan}}, \bibinfo {author} {\bibfnamefont {M.}~\bibnamefont
  {Lewenstein}}, \ and\ \bibinfo {author} {\bibfnamefont {A.}~\bibnamefont
  {Riera}},\ }\href@noop {} {\bibfield  {journal} {\bibinfo  {journal} {New J.
  Phys.}\ }\textbf {\bibinfo {volume} {20}},\ \bibinfo {pages} {033032}
  (\bibinfo {year} {2018})},\ \Eprint {http://arxiv.org/abs/1704.06646}
  {arXiv:1704.06646} \BibitemShut {NoStop}%
\bibitem [{\citenamefont {Wilming}\ \emph {et~al.}(2017)\citenamefont
  {Wilming}, \citenamefont {Goihl}, \citenamefont {Krumnow},\ and\
  \citenamefont {Eisert}}]{Wilming2017}%
  \BibitemOpen
  \bibfield  {author} {\bibinfo {author} {\bibfnamefont {H.}~\bibnamefont
  {Wilming}}, \bibinfo {author} {\bibfnamefont {M.}~\bibnamefont {Goihl}},
  \bibinfo {author} {\bibfnamefont {C.}~\bibnamefont {Krumnow}}, \ and\
  \bibinfo {author} {\bibfnamefont {J.}~\bibnamefont {Eisert}},\ }\href
  {http://arxiv.org/abs/1704.06291} {\enquote {\bibinfo {title} {{Towards local
  equilibration in closed interacting quantum many-body systems}},}\ }
  (\bibinfo {year} {2017}),\ \Eprint {http://arxiv.org/abs/1704.06291}
  {1704.06291} \BibitemShut {NoStop}%
\bibitem [{\citenamefont {Uhlmann}(1976)}]{Uhlmann1976}%
  \BibitemOpen
  \bibfield  {author} {\bibinfo {author} {\bibfnamefont {A.}~\bibnamefont
  {Uhlmann}},\ }\href {\doibase 10.1016/0034-4877(76)90060-4} {\bibfield
  {journal} {\bibinfo  {journal} {Rep. Math. Phys.}\ }\textbf {\bibinfo
  {volume} {9}},\ \bibinfo {pages} {273} (\bibinfo {year} {1976})}\BibitemShut
  {NoStop}%
\bibitem [{\citenamefont {Nielsen}\ and\ \citenamefont
  {Chuang}(2000)}]{Nielsen2000}%
  \BibitemOpen
  \bibfield  {author} {\bibinfo {author} {\bibfnamefont {M.~A.}\ \bibnamefont
  {Nielsen}}\ and\ \bibinfo {author} {\bibfnamefont {I.~L.}\ \bibnamefont
  {Chuang}},\ }\href@noop {} {\emph {\bibinfo {title} {Quantum computation and
  quantum information}}}\ (\bibinfo  {publisher} {Cambridge University Press},\
  \bibinfo {year} {2000})\BibitemShut {NoStop}%
\bibitem [{\citenamefont {Wolf}\ \emph {et~al.}(2008)\citenamefont {Wolf},
  \citenamefont {Verstraete}, \citenamefont {Hastings},\ and\ \citenamefont
  {Cirac}}]{Wolf2008}%
  \BibitemOpen
  \bibfield  {author} {\bibinfo {author} {\bibfnamefont {M.~M.}\ \bibnamefont
  {Wolf}}, \bibinfo {author} {\bibfnamefont {F.}~\bibnamefont {Verstraete}},
  \bibinfo {author} {\bibfnamefont {M.~B.}\ \bibnamefont {Hastings}}, \ and\
  \bibinfo {author} {\bibfnamefont {J.~I.}\ \bibnamefont {Cirac}},\ }\href
  {\doibase 10.1103/PhysRevLett.100.070502} {\bibfield  {journal} {\bibinfo
  {journal} {Phys. Rev. Lett.}\ }\textbf {\bibinfo {volume} {100}},\ \bibinfo
  {pages} {070502} (\bibinfo {year} {2008})}\BibitemShut {NoStop}%
\bibitem [{\citenamefont {Fannes}\ \emph {et~al.}(1992)\citenamefont {Fannes},
  \citenamefont {Nachtergaele},\ and\ \citenamefont {Werner}}]{Fannes1992}%
  \BibitemOpen
  \bibfield  {author} {\bibinfo {author} {\bibfnamefont {M.}~\bibnamefont
  {Fannes}}, \bibinfo {author} {\bibfnamefont {B.}~\bibnamefont
  {Nachtergaele}}, \ and\ \bibinfo {author} {\bibfnamefont {R.~F.}\
  \bibnamefont {Werner}},\ }\href {\doibase 10.1007/bf02099178} {\bibfield
  {journal} {\bibinfo  {journal} {Commun. Math. Phys.}\ }\textbf {\bibinfo
  {volume} {144}},\ \bibinfo {pages} {443} (\bibinfo {year}
  {1992})}\BibitemShut {NoStop}%
\bibitem [{\citenamefont {Nachtergaele}(1996)}]{Nachtergaele1996}%
  \BibitemOpen
  \bibfield  {author} {\bibinfo {author} {\bibfnamefont {B.}~\bibnamefont
  {Nachtergaele}},\ }\href {\doibase 10.1007/bf02099509} {\bibfield  {journal}
  {\bibinfo  {journal} {Commun. Math. Phys.}\ }\textbf {\bibinfo {volume}
  {175}},\ \bibinfo {pages} {565} (\bibinfo {year} {1996})}\BibitemShut
  {NoStop}%
\bibitem [{\citenamefont {Perez-Garcia}\ \emph {et~al.}(2006)\citenamefont
  {Perez-Garcia}, \citenamefont {Verstraete}, \citenamefont {Wolf},\ and\
  \citenamefont {Cirac}}]{Perez-Garcia2006}%
  \BibitemOpen
  \bibfield  {author} {\bibinfo {author} {\bibfnamefont {D.}~\bibnamefont
  {Perez-Garcia}}, \bibinfo {author} {\bibfnamefont {F.}~\bibnamefont
  {Verstraete}}, \bibinfo {author} {\bibfnamefont {M.~M.}\ \bibnamefont
  {Wolf}}, \ and\ \bibinfo {author} {\bibfnamefont {J.~I.}\ \bibnamefont
  {Cirac}},\ }\href@noop {} {\bibfield  {journal} {\bibinfo  {journal} {Quant.
  Inf. Comp.}\ }\textbf {\bibinfo {volume} {5$\&$6}},\ \bibinfo {pages} {401}
  (\bibinfo {year} {2006})}\BibitemShut {NoStop}%
\bibitem [{\citenamefont {Perez-Garcia}\ \emph {et~al.}(2008)\citenamefont
  {Perez-Garcia}, \citenamefont {Verstraete}, \citenamefont {Wolf},\ and\
  \citenamefont {Cirac}}]{Perez-Garcia2007}%
  \BibitemOpen
  \bibfield  {author} {\bibinfo {author} {\bibfnamefont {D.}~\bibnamefont
  {Perez-Garcia}}, \bibinfo {author} {\bibfnamefont {F.}~\bibnamefont
  {Verstraete}}, \bibinfo {author} {\bibfnamefont {M.~M.}\ \bibnamefont
  {Wolf}}, \ and\ \bibinfo {author} {\bibfnamefont {J.~I.}\ \bibnamefont
  {Cirac}},\ }\href {http://dl.acm.org/citation.cfm?id=2016976.2016982}
  {\bibfield  {journal} {\bibinfo  {journal} {Quant. Info. Comput.}\ }\textbf
  {\bibinfo {volume} {8}},\ \bibinfo {pages} {650} (\bibinfo {year}
  {2008})}\BibitemShut {NoStop}%
\bibitem [{\citenamefont {Chen}\ \emph {et~al.}(2010)\citenamefont {Chen},
  \citenamefont {Gu},\ and\ \citenamefont {Wen}}]{PhysRevB.82.155138}%
  \BibitemOpen
  \bibfield  {author} {\bibinfo {author} {\bibfnamefont {X.}~\bibnamefont
  {Chen}}, \bibinfo {author} {\bibfnamefont {Z.-C.}\ \bibnamefont {Gu}}, \ and\
  \bibinfo {author} {\bibfnamefont {X.-G.}\ \bibnamefont {Wen}},\ }\href
  {\doibase 10.1103/PhysRevB.82.155138} {\bibfield  {journal} {\bibinfo
  {journal} {Phys. Rev. B}\ }\textbf {\bibinfo {volume} {82}},\ \bibinfo
  {pages} {155138} (\bibinfo {year} {2010})}\BibitemShut {NoStop}%
\bibitem [{\citenamefont {Van~Acoleyen}\ \emph {et~al.}(2013)\citenamefont
  {Van~Acoleyen}, \citenamefont {Mari\"en},\ and\ \citenamefont
  {Verstraete}}]{PhysRevLett.111.170501}%
  \BibitemOpen
  \bibfield  {author} {\bibinfo {author} {\bibfnamefont {K.}~\bibnamefont
  {Van~Acoleyen}}, \bibinfo {author} {\bibfnamefont {M.}~\bibnamefont
  {Mari\"en}}, \ and\ \bibinfo {author} {\bibfnamefont {F.}~\bibnamefont
  {Verstraete}},\ }\href {\doibase 10.1103/PhysRevLett.111.170501} {\bibfield
  {journal} {\bibinfo  {journal} {Phys. Rev. Lett.}\ }\textbf {\bibinfo
  {volume} {111}},\ \bibinfo {pages} {170501} (\bibinfo {year}
  {2013})}\BibitemShut {NoStop}%
\bibitem [{\citenamefont {Kliesch}\ \emph {et~al.}(2014)\citenamefont
  {Kliesch}, \citenamefont {Gogolin}, \citenamefont {Kastoryano}, \citenamefont
  {Riera},\ and\ \citenamefont {Eisert}}]{Kliesch2014}%
  \BibitemOpen
  \bibfield  {author} {\bibinfo {author} {\bibfnamefont {M.}~\bibnamefont
  {Kliesch}}, \bibinfo {author} {\bibfnamefont {C.}~\bibnamefont {Gogolin}},
  \bibinfo {author} {\bibfnamefont {M.~J.}\ \bibnamefont {Kastoryano}},
  \bibinfo {author} {\bibfnamefont {A.}~\bibnamefont {Riera}}, \ and\ \bibinfo
  {author} {\bibfnamefont {J.}~\bibnamefont {Eisert}},\ }\href@noop {}
  {\bibfield  {journal} {\bibinfo  {journal} {Phys. Rev. X}\ }\textbf {\bibinfo
  {volume} {4}},\ \bibinfo {pages} {031019} (\bibinfo {year}
  {2014})}\BibitemShut {NoStop}%
\end{thebibliography}
\end{document}